\documentclass{article} % for Combinatorica
\usepackage{fullpage}

\usepackage{subfigure}
\usepackage{graphicx}
\usepackage{color}
\usepackage{enumerate, verbatim}
\usepackage{algorithm,algorithmic} 
\usepackage{amssymb, amsmath, amsthm}

\newtheorem{lemma}{Lemma}
\newtheorem{theorem}{Theorem}

\newtheorem{proposition}{Proposition}

\newcommand{\STAB}{\mathrm{STAB}}
\newcommand{\ent}{H}

\title{Sorting under Partial Information\\
(without the Ellipsoid Algorithm)\footnote{This work was supported by the ``Actions de Recherche Concert\'ees'' (ARC) fund of the ``Communaut\'e fran\c{c}aise de Belgique'', NSERC of Canada, and the Canada Research Chairs Programme. G.J.\ and R.J.\ are Postdoctoral Researchers of the ``Fonds National de la Recherche Scientifique'' (F.R.S.--FNRS). A preliminary version of the work appeared in \cite{SUPI-STOC}.}}

\date{}

\author{Jean Cardinal, Samuel Fiorini, Gwena\"el Joret\footnote{Universit\'e Libre de Bruxelles (ULB), Brussels, Belgium. {E-mail: \tt\small \{jcardin,sfiorini,gjoret\}@ulb.ac.be}.},\\ 
Rapha\"el M. Jungers\footnote{Universit\'e Catholique de Louvain (UCL), Louvain-La-Neuve, Belgium. {E-mail: \tt\small raphael.jungers@uclouvain.be}}, J. Ian Munro%
\footnote{University of Waterloo, Waterloo, Ontario, Canada. {E-mail: \tt\small imunro@uwaterloo.ca}}
}

\begin{document}

\maketitle

\sloppy

\begin{abstract}
We revisit the well-known problem of sorting under partial information: sort a finite set given the outcomes of comparisons between some pairs of elements. The input is a partially ordered set $P$, and solving the problem amounts to discovering an unknown linear extension of $P$, using pairwise comparisons. The information-theoretic lower bound on the number of comparisons needed in the worst case is $\log e(P)$, the binary logarithm of the number of linear extensions of $P$. In a breakthrough paper, Jeff Kahn and Jeong Han Kim ({\em J.\ Comput.\ System Sci.\ 51 (3), 390--399, 1995}) showed that there exists a polynomial-time algorithm for the problem achieving this bound up to a constant factor. Their algorithm invokes the ellipsoid algorithm at each iteration for determining the next comparison, making it impractical.

We develop efficient algorithms for sorting under partial information. Like Kahn and Kim, our approach relies on graph entropy. However, our algorithms differ in essential ways from theirs. Rather than resorting to convex programming for computing the entropy, we approximate the entropy, or make sure it is computed only once, in a restricted class of graphs, permitting the use of a simpler algorithm. Specifically, we present:
\begin{enumerate}
\item an $O(n^2)$ algorithm performing $O(\log n\cdot \log e(P))$ comparisons;
\item an $O(n^{2.5})$ algorithm performing at most $(1+\varepsilon) \log e(P) + O_\varepsilon(n)$ comparisons;
\item an $O(n^{2.5})$ algorithm performing $O(\log e(P))$ comparisons.
\end{enumerate}
All our algorithms can be implemented in such a way that their computational bottleneck is confined in a preprocessing phase, while the sorting phase is completed in $O(q) + O(n)$ time, where $q$ denotes the number of comparisons performed.
\end{abstract}

\newpage
\section{Introduction}

\paragraph*{Problem Definition}

We consider the following problem:\medskip

{\it Let $V = \{ v_1 , \ldots, v_n \}$ be a set equipped with an unknown linear order $\leqslant$. Given a subset of the relations $v_i \leqslant v_j$, determine the complete linear order by queries of the form: ``is $v_i \leqslant v_j$?''.}\medskip

\begin{figure}
\begin{center}
\includegraphics[scale=.5]{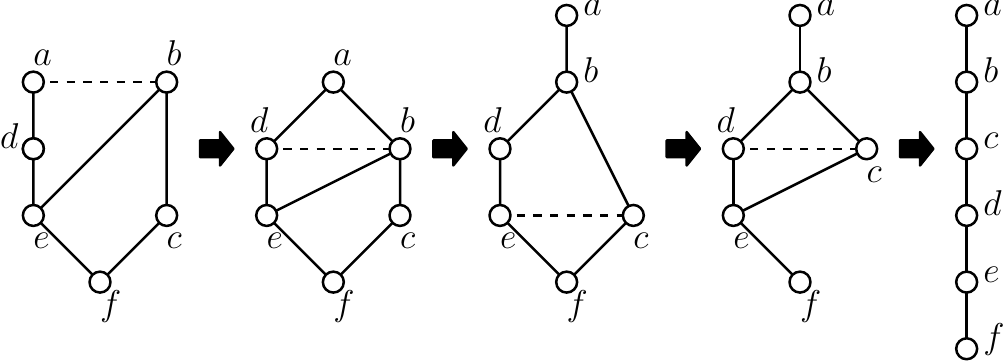}
\end{center}
\caption{\label{fig:example}An instance of the problem of sorting under partial information. In this example, we use 4 comparisons (dashed edges). At every step, the Hasse diagram of the currently known partial order is shown.}
\end{figure}

This problem is called {\it Sorting under Partial Information}. We are given the outcomes of a number of comparisons between elements of a linearly ordered set, and we wish to ``complete the sort'' by performing more comparisons. The partially ordered set (poset) $P = (V,\leqslant_P)$ encoding these known outcomes is a partial information that should help reducing the number of comparisons performed. Denoting by $e(P)$ the number of linear extensions of $P$, it is obvious that the number of required comparisons is at least $\log e(P)$ in the worst case\footnote{Throughout the paper, $\log x$ denotes the binary logarithm of $x$.}. An example is given in Figure~\ref{fig:example}.

\paragraph*{Previous Results}

The problem was first posed by Fredman~\cite{F76}. He showed that there exists an algorithm that performs $\log e(P) + 2n$ additional comparisons between elements of $V$. However, the number of comparisons performed by Fredman's algorithm is not $O(\log e(P))$ when $\log e(P)$ is sub-linear, and deciding what comparisons should be done takes super-polynomial time. At that time, it remained open whether there existed, on the one hand, an algorithm performing $O(\log e(P))$ comparisons, and, on the other hand, an algorithm running in polynomial time.

The first question was answered by Kahn and Saks~\cite{KS84j}. They showed that there always exists a query of the form ``is $v_i \leqslant v_j$?'' such that the fraction of linear extensions in which $v_i$ is smaller than $v_j$ lies in the interval $(3/11,8/11)$. This is a relaxation of the well-known $1/3$--$2/3$ conjecture, a conjecture formulated independently by Fredman, Linial, and Stanley, see \cite{L84}. A simpler proof yielding weaker bounds was given by Kahn and Linial~\cite{KL91}. Better bounds were later given by Brightwell, Felsner, and Trotter~\cite{BFT95}, and Brightwell~\cite{B99}. Iteratively choosing such a comparison yields an algorithm that performs $O(\log e(P))$ comparisons. However, finding the right comparisons remained intractable.

In 1995, Kahn and Kim published a breakthrough paper~\cite{KK95} in which they describe a polynomial-time algorithm performing $O(\log e(P))$ comparisons, thus answering both questions positively. Their key insight is to relate $\log e(P)$ to the entropy of the incomparability graph of $P$, a quantity that can be computed in polynomial time. Their algorithm, although polynomial, is still far from practical because it uses the ellipsoid algorithm $O(\log e(P)) = O(n \log n)$ times to determine the comparisons.

\paragraph*{Contribution}

Our results are summarized in Table \ref{tab:results} below. 

\begin{table}[h!]
\begin{center}
\begin{tabular}{|c|c|c|}
\hline
Algorithm & Global complexity & Number of comparisons\\
\hline
\hline
\cite{KK95} & $O(n \log n \cdot EA(n))$& $\le 9.82\cdot\log e(P)$\\
\hline
\hline
{\bf Algorithm 1} & $O(n^2)$ & $O(\log n \cdot \log e(P))$\\
\hline
{\bf Algorithm 2} & $O(n^{2.5})$ & $\le (1+\varepsilon) \log e(P) + O_\varepsilon(n)$\\
\hline
{\bf Algorithm 3} & $O(n^{2.5})$ & $\le 15.09 \cdot \log e(P)$\\
\hline
\end{tabular}
\end{center}
\label{tab:results}
\caption{We denote by $EA(n)$ the time needed for the ellipsoid algorithm to compute the entropy of a poset of order $n$. The original bound given by Kahn and Kim on the number of comparisons performed by their algorithm is $54.45 \cdot \log e(P)$. The improved bound given in the table is a byproduct of our results. (The notation $O_\varepsilon(n)$ means that the hidden constant may depend on $\varepsilon$.)
}
\end{table}

We now compare these results to those of Kahn and Kim (denoted: K\&K). In terms of global complexity, each of our algorithms greatly improves over that of K\&K. Furthermore:

\begin{itemize}
\item If $\log e(P)$ is super-linear in $n$, the number of comparisons of our second algorithm is lower than that of K\&K. By optimizing over $\varepsilon$, it can be shown that the number of comparisons is actually $\log e(P) + o(\log e(P)) + O(n)$ in this case, a number of comparisons comparable to that of Fredman's algorithm. 

\item If $\log e(P)$ is linear or sub-linear in $n$, the number of comparisons of our third algorithm is comparable to that of K\&K, although the constant in front of $\log e(P)$ is still far from the best constant achieved by a super-polynomial algorithm via balancing pairs \cite{BFT95,B99}.

\item Our algorithms have the following useful property: they compute information that guides the sorting and can then be reused to solve any given instance with the same partial information $P$, in time proportional to the number of comparisons, plus a term linear in $n$.
\end{itemize}

Finally, note that randomized algorithms for sorting under partial information can be derived from random linear extension generation algorithms. The idea here would be to estimate the efficiency of a comparison -- that is, the fraction of linear extensions remaining after some query ``is $v_i \leqslant v_j$?'' is performed -- arbitrarily closely by testing a sufficiently large random sample of linear extensions. However, the running time of such an algorithm would be much higher than the ones we propose here. For instance the recent sampling algorithm from Huber~\cite{H06}, has expected running time $O(n^3\log n)$, and this sampling step has to be performed a large number of times.

\paragraph*{Outline and Key Ideas}

K\&K showed that graph entropy, as defined by K\"orner~\cite{K73}, is a useful tool in the problem of sorting under partial information. Letting $H(\bar{P})$ be the entropy of the incomparability graph of $P$, they showed that $\log e(P) = \Theta(nH(\bar{P}))$. Every comparison performed by their algorithm decreases $n H(\bar{P})$ by at least some constant. Hence the total number of comparisons is $O(nH(\bar{P}))$ and thus $O(\log e(P))$. Furthermore, their algorithm is polynomial, because the entropy can be computed in polynomial time using convex programming.

Our goal is to obtain practical algorithms, without sacrificing the number of comparisons. Our first key idea is to compute a {\sl greedy chain decomposition\/} of $P$, that is, a partition of $P$ into chains (totally ordered subsets), obtained by iteratively extracting a longest chain. This allows us to get rid of the costly convex programming machinery and enables us to focus only on the relevant part of $P$. In \cite{POP_SICOMP}, we have provided bounds on the amount of information (in terms of entropy) that is lost when we forget the relations of $P$ between two distinct chains of a greedy chain decomposition.

As a warmup, we first describe Algorithm 1, an insertion sort-like algorithm. Then we describe Algorithm 2, a mergesort-like algorithm: find a greedy chain decomposition of $P$, and merge the chains using a simple linear-time merging algorithm. The number of comparisons performed by this algorithm can be shown to be close to $\log e(P)$, up to an arbitrarily small factor and a term linear in $n$. This is described in Section~\ref{sec:merge}.

As noted above, our mergesort-like algorithm performs better than that of K\&K provided the information theoretic lower bound $\log e(P)$ is super-linear. The algorithms are comparable (in terms of number of comparisons) if $\log e(P)$ is linear. If $\log e(P)$ is sub-linear, we have to use another strategy: instead of forgetting all the relations of $P$ between the chains of a greedy chain decomposition, we keep some of them. Namely, we keep all the relations between the elements of the longest chain and the rest of $P$. When $\log e(P)$ is small compared to $n$, the longest chain contains a large fraction of the elements. Hence, this less radical strategy keeps most of the information contained in $P$.

Our second key idea is contained in the following algorithm: find a longest chain $A$, use the mergesort-like algorithm on $P-A$, yielding a chain $B$, and cautiously merge the chains $A$ and $B$ using the current partial information. Thus we reduce the general sorting problem to an easier subproblem known as {\sl merging under partial information}. It is a special case of the problem of sorting under partial information in which $P$ can be covered by exactly two chains, and has been studied by Linial~\cite{L84}. By using an algorithm for merging under partial information performing $O(\log e(P))$ comparisons, we obtain an algorithm for the general sorting problem performing $O(\log e(P))$ comparisons. This is shown in Section \ref{sec:cautious_merge}.

The problem of merging under partial information is tackled in Section~\ref{sec:MUPIsec}. Linial~\cite{L84} already provided an algorithm for the problem, but we develop an alternative solution. We first show that in this special case, the entropy of the incomparability graph of $P$ can be computed very easily. The computation relies on a structural lemma on the entropy of bipartite graphs by K\"orner and Marton~\cite{KM88}, and on the additional structure exhibited by the incomparability graph of a poset covered by two chains.

Then, we show that given the vertex weights achieving the entropy, there exists a sequence of pairwise chain mergings, each of which decreases $nH(\bar{P})$ by an amount proportional to the number of comparisons performed. After each merging, the weights on the vertices can be updated efficiently. This yields the desired algorithm for merging under partial information, and thus an algorithm for sorting under partial information performing $O(\log e(P))$ comparisons. We refer to it as Algorithm 3. The global complexity of Algorithm 3 is $O(n^{2.5})$.

The plan of the paper is as follows. Preliminaries on complexity measures, the entropy of a graph, and greedy chain decompositions, are given in Section~\ref{sec:prelim}. In Section~\ref{sec:tight_bound}, we offer new results on the entropy, improving several aspects of K\&K's analysis. Mainly, we prove the tight inequality $nH(\bar{P}) \leq 2 \log e(P)$, whereas K\&K show $nH(\bar{P}) \leq (1+7\log e) \log e(P) \simeq 11.1 \log e(P)$.

As a first simple example of a near-optimal algorithm for sorting under partial information, we describe our (simple) Algorithm~1 in Section~\ref{sec:insertion}. This algorithm has global complexity $O(n^2)$ and performs a number of comparisons within a $\log n$ factor only of the information-theoretic lower bound. 

As mentioned above, the mergesort-like algorithm (Algorithm 2) is given in Section~\ref{sec:merge},
while Sections \ref{sec:cautious_merge} and~\ref{sec:MUPIsec} are devoted to
Algorithm 3 performing $O(\log e(P))$ comparisons. In the last section, Section~\ref{sec:reuse}, we explain how that algorithm can be implemented in such a way that all costly computations are done in a {\em preprocessing phase}. As a result, the algorithm can reuse the information computed during that preprocessing phase and solve any other instance with the same partial information $P$, in time proportional to the number of comparisons plus a term linear in $n$.

As a final remark, we report an important observation from an anonymous referee concerning Linial's algorithm for merging under partial information~\cite{L84}. Using dynamic programming, it can be shown that this algorithm can be implemented in a way that would be competitive with our proposition. It would not, however, have a sorting phase that is as efficient.

A related note is that the algorithm for merging under partial information given in the preliminary version~\cite{SUPI-STOC} of this paper is slightly different from the one presented here. The resulting new algorithm for sorting under partial information is simpler and can be implemented so that the sorting phase takes $O(q) + O(n)$ time, where $q$ is the number of comparisons performed by the algorithm. Achieving the latter property was left as an open problem in~\cite{SUPI-STOC}.

We also include an appendix, in which we discuss the complexities of some important steps used in our algorithms, among which is the construction of a greedy chain decomposition.

\paragraph*{Other Related Works}
In 2004, Yao proved that the information-theoretic lower bound for the problem of sorting under partial information also holds for quantum decision trees, up to a term linear in $n$~\cite{Yao04}. His analysis also relies on the notion of graph entropy.

In a recent paper, Daskalakis et al.~\cite{DKMRV09} analyze the problem of discovering a partial order using comparisons. In that setting, a comparison can have three outcomes, including one stating that the two elements are incomparable, and the goal is to completely identify the underlying partial order. They propose an algorithm performing a number of comparisons that is within a constant factor of the information-theoretic lower bound for partial orders of a given width.

\section{Preliminaries}
\label{sec:prelim}

We give a number of definitions and basic results, and summarize the contribution of Kahn and Kim~\cite{KK95} to the problem. 

\paragraph*{Complexity Measures}

Consider an algorithm for sorting under partial information. The {\sl query complexity\/} is the number of comparisons between elements of $P$ that are done by the algorithm. The {\sl preprocessing complexity\/} measures the computational work done before the first comparison is performed. The rest of the work is measured by the {\sl sorting complexity}. The {\sl preprocessing phase\/} and {\sl sorting phases\/} are defined similarly. Thus, in the preprocessing phase, we are restricted to only process the input poset. The comparisons are performed during the sorting phase. The {\sl global complexity\/} is simply the sum of the preprocessing and sorting complexities.

Our model of computation is a RAM machine with $\Theta (\log n)$-size words. The global complexity is measured as the total number of arithmetic and logical operations on words.

\paragraph*{Entropy and Sorting}

We recall that a subset $S$ of vertices of a graph is a {\sl stable set\/} (or {\sl independent set}) if the vertices in $S$ are pairwise nonadjacent. The {\sl stable set polytope\/} of a graph $G$ with vertex set $V$ and order $n$ is the $n$-dimensional polytope
\begin{equation*}
\STAB(G) := \mathrm{\ conv} \{\chi^S \in \mathbb{R}^V : S \textrm{ stable set in }G\},
\end{equation*}
where $\chi^S$ is the characteristic vector of the subset $S$, assigning the value $1$ to every vertex in $S$, and $0$ to the others. The {\sl entropy\/} of $G$ is defined as (see~\cite{K73,CKLMS90})
\begin{equation}
\label{def-H}
H(G) := \min_{x\in \STAB(G)} - \frac 1n \sum_{v\in V} \log x_v.
\end{equation}
Any point $x \in \STAB(G)$ describes a feasible solution of the convex program defined in the right-hand side of \eqref{def-H}. The {\sl entropy\/} of $x$ is the value of the objective function of that program with respect to $x$, which we denote by $\ent(x)$.

For any given poset $P$, we consider two graphs: the {\sl comparability graph\/} $G(P)$ and the {\sl incomparability graph\/} $\bar{G}(P)$. The vertex set of $G(P)$ is the ground set of $P$ and two distinct vertices $v$ and $w$ are adjacent in $G(P)$ whenever they are comparable in $P$. The incomparability graph $\bar{G}(P)$ is simply the complement of $G(P)$. Following K\&K, we denote by $H(P)$ the entropy of $G(P)$ and by $H(\bar{P})$ the entropy of $\bar{G}(P)$.

Entropy plays an important role in the sorting under partial information problem. 

The first reason is explained by the following result due to K\&K. In particular, it implies $\log e(P) = \Theta(nH(\bar{P}))$. Thus the information theoretic lower bound and the entropy of the incomparability graph of $P$ are tightly related.

\begin{lemma}[\cite{KK95}]
\label{lem:kk}
For any poset $P$ of order $n$,
$\log e(P) \leq nH(\bar{P}) \leq \min \{\log e(P) + \log e \cdot n, c_1 \log e(P)\},$
where $c_1 = (1+7\log e) \simeq 11.1$.
\end{lemma}

The second reason is that, while computing $e(P)$ is $\#P$-complete \cite{BW91}, computing $H(\bar{P})$ can be done in polynomial time by solving the convex minimization problem (\ref{def-H}), as we now explain. When $G = \bar{G}(P)$, the stable set polytope $\STAB(G)$ has a known description in terms of linear inequalities. Although the number of inequalities is (in most cases) exponential, the corresponding separation problem can be solved efficiently. Hence (\ref{def-H}) can be solved by the ellipsoid algorithm. 
(To be precise, the ellipsoid algorithm will actually {\em approximate}
the optimum of (\ref{def-H}) to any fixed precision, in polynomial time.)

Much of this favorable behaviour is due to the perfection of $\bar{G}(P)$. We recall that a graph $G$ is {\sl perfect\/} if $\omega(H)=\chi(H)$ holds for every induced subgraph $H$ of $G$, where $\omega(H)$ and $\chi(H)$ denote the clique and chromatic numbers of $H$, respectively. If $G$ is perfect, then its complement $\bar{G}$ is also perfect~\cite{L72}. It is known that the comparability graph $G(P)$ of $P$ is perfect, and therefore so is the incomparability graph $\bar{G}(P)$ of $P$. The latter statement is known as Dilworth's Theorem. The following basic result is a manifestation of convex programming duality (see for instance \cite{S95} for a proof).

\begin{lemma}
\label{lem:perf}
Assume $G$ is a perfect graph with vertex set $V$ and order $n$, and let $x \in \mathbb{R}^V$ and $z \in \mathbb{R}^V$ be feasible solutions to (\ref{def-H}) for $G$ and $\bar{G}$, respectively. Then $x$ and $z$ are optimal iff $x_v z_v = 1/n$ for all $v \in V$. In particular, $H(G) + H(\bar G)=\log n$.
\end{lemma}

Csisz{\'a}r et al.~\cite{CKLMS90} have characterized perfect graphs as the graphs that ``split graph entropy''. More precisely, they proved that $G$ is perfect if and only if, for {\em every\/} probability distribution $p$ on the vertex set of $G$, the sum of the entropies of $G$ and $\bar{G}$ with respect to $p$ (see the references for a precise definition of this) equals the (Shannon) entropy of $p$.

The algorithm of Kahn and Kim~\cite{KK95} is based on two main lemmas, Lemma \ref{lem:kk} above and the next lemma. Whenever $a$ and $b$ are incomparable elements of $P$, we denote by $P(a < b)$ the poset obtained by adding the relation $(a,b)$ to the partial order of $P$ and then closing transitively.

\begin{lemma}[\cite{KK95}]
\label{lem:kk2}
In any poset $P$ of order $n$ that is not a chain there are $a$, $b$ incomparable such that
$$
\max \{nH(\overline{P(a<b)}),nH(\overline{P(b<a)})\} \leq nH(\bar{P}) - c_2,
$$
where $c_2 = \log (1 + 17/112) \simeq 0.2$.
\end{lemma}

\paragraph*{The Algorithm of K\&K and its Complexity}

Let $V$ denote the ground set of $P$. Given an optimal solution $x \in \mathbb{R}^V$ to (\ref{def-H}) for $G(P)$, K\&K show how to choose a pair $a$, $b$ as in Lemma \ref{lem:kk2}. Knowing the {\sl primal\/} solution $x$, this choice can be done efficiently (in $O(n^2)$ time).

Comparing $a$ and $b$ gives a new partial information $P' \in \{P(a<b),P(b<a)\}$. The key is that for any outcome, $nH(\bar{P'}) \leq nH(\bar{P})-c_2$. This is proved by modifying appropriately an optimal {\sl dual\/} solution, that is, an optimal solution $z \in \mathbb{R}^V$ to (\ref{def-H}) for $\bar{G}(P)$. By Lemma \ref{lem:perf}, $z_v = 1/(nx_v)$ for all $v \in V$. Knowing $x$, a new dual solution $z'$ can be efficiently constructed (in $O(n^3)$ time).

To determine the next comparison, the K\&K algorithm needs to compute an optimal solution $x'$ to (\ref{def-H}) for $G(P')$. Because the optimality of $z'$ is not guaranteed, letting $x'_v = 1/(nz'_v)$ for $v \in V$ does not work. This explains why their algorithm uses the ellipsoid algorithm before each comparison.

We have shown in \cite{POP_SICOMP} that $H(P)$ can be expressed via a convex minimization problem with $2n$ variables and at most $n^2$ constraints, making possible the use of interior point algorithms for computing $H(P)$ (this alternative formulation is described in Section~\ref{sec:tight_bound}).
Although this makes the K\&K algorithm more practical, this does not make it competitive with our algorithms in terms of running time since it is unlikely that computing $H(P)$ using interior point algorithms can be done in less than $O(n^4)$ time (plugging in in a straightforward way the number of variables and constraints in complexity bounds for interior point algorithms would yield a $O(n^6)$ complexity \cite{BV04}).

\paragraph*{Greedy Chain Decompositions}

Suppose we want to approximate the entropy $H(G)$ of a given perfect graph $G$. 
We have shown \cite{POP_SICOMP} that the following greedy heuristic performs very well. First, iteratively remove a maximum stable set in $G$. Denote by $S_1$, \ldots, $S_k$ the stable sets extracted from $G$. Second, construct the {\sl greedy point\/}
$$
x := \sum_{i=1}^k \frac{|S_i|}{n} \chi^{S_i}
$$ 
in $\STAB(G)$. The entropy of this point is
$$
\ent(x) = -\frac{1}{n} \sum_{v \in V} \log x_v
= \sum_{i = 1}^k -\frac{|S_i|}{n} \log \frac{|S_i|}{n}.
$$
Note that this is precisely the entropy of the probability distribution $\big\{\frac{|S_1|}{n},\ldots,\frac{|S_k|}{n}\big\}$.

\begin{theorem}[\cite{POP_SICOMP}]
\label{thm:greedy}
Let $G$ be a perfect graph on $n$ vertices and let $x$ be an arbitrary greedy point in $\STAB(G)$. 
Then, for every $\varepsilon > 0$,
$$
\ent(x) \leq (1+\varepsilon) H(G) + (1+\varepsilon) \log \Big(1 + \frac{1}{\varepsilon}\Big).
$$
\end{theorem}

In the context of the sorting under partial information problem, we apply the greedy heuristic to $\bar{G}(P)$. This gives a decomposition of $P$ into chains $C_1$, \ldots, $C_k$ that we call a {\sl greedy chain decomposition}. Although the fastest known algorithm for computing a maximum chain in a poset of order $n$ has complexity $O(n^2)$ (see \cite{Gbook}, Chapter 5), a greedy chain decomposition can be found in $O(n^{2.5})$ time, see Appendix~\ref{app:greedy}.

\section{A Tight Bound on the Entropy of an Incomparability Graph}
\label{sec:tight_bound}

K\&K conjectured that the value for the constant $c_1$ in Lemma~\ref{lem:kk} could be improved to $c_{1} = 1 + \log e \simeq 2.44$. We show that one can actually take $c_{1}=2$, which is best possible, as shown by the poset consisting of two incomparable elements.

\begin{theorem}
\label{th:c1}
For any poset $P$ of order $n$,
$$
n H(\bar{P}) \leq 2 \log e(P).
$$
\end{theorem}

Before proving this result, we give an equivalent definition of the entropy of a poset in terms of consistent collections of intervals, that is used crucially in our proof of Theorem~\ref{th:c1}.

We say that a collection of open intervals $\{(y_{v^-},y_{v^+})\}_{v \in V}$, each of which is contained in the interval $(0,1)$, is {\sl consistent\/} with $P$ if $v <_P w$ implies that the interval for $v$ is entirely to the left of the interval for $w$, that is, $y_{v^+} \leq y_{w^-}$. We denote $\mathcal{I}(P)$ the set of all such collections of intervals.

As is easily seen \cite{POP_SICOMP}, $H(P)$ equals the minimum of 
$$- \frac{1}{n} \sum_{v \in V} \log x_v$$
over all vectors $x \in \mathbb{R}_+^V$ such that there exists a collection of intervals in $\mathcal{I}(P)$ where, for each $v \in V$, the interval for $v$ has length $x_v$. In other words, the following lemma holds.

\begin{lemma}
\label{lem:intH}
Let $P$ be a poset of order $n$ with ground set $V$. Then, we have
$$
H(P) = \min \Big\{ - \frac{1}{n} \sum_{v \in V} \log x_v : \exists \{(y_{v^-},y_{v^+})\}_{v \in V} \in \mathcal{I}(P) \textrm{ s.t. } \forall v \in V : x_v = y_{v^+} - y_{v^-}\Big\}.
$$
\end{lemma}

This new definition of the entropy of a poset has some advantages.

First, it yields a convex program with $2n$ variables and at most $n^2$ constraints for computing the entropy. This shows that the entropy of a poset can be computed with interior point algorithms.

Second, it gives a more intuitive framework to reason about the entropy of a poset. As an illustration we sketch short proofs of two results by Kahn and Kim~\cite{KK95}.

In order to show that $\log e(P) \leq nH(\bar{P}) \leq \log e(P) + \log e \cdot n$, the ``easy part'' of Lemma \ref{lem:kk}, K\&K consider an optimal solution $x$ to \eqref{def-H}. Because $x$ is feasible, it defines a box that is contained in $\STAB(G(P))$. (The defining property of this box is that it has $x$ and the origin as opposite vertices.) Because $x$ is optimal, it yields a simplex that contains $\STAB(G(P))$. Thus the box is contained in $\STAB(G)$, which is contained in the simplex. This gives inequalities between the volumes of these polytopes. The volume of the box is $2^{-nH(P)}$ and that of the simplex is $(n^n/n!)\,2^{-nH(P)}$. By invoking a beautiful result of Stanley~\cite{S86} relating the volume of $\STAB(G(P))$ to $e(P)$, and also $H(P) + H(\bar{P}) = \log n$ (see Lemma \ref{lem:perf}), K\&K derive the desired inequalities. Stanley~\cite{S86} proves that the stable set polytope $\STAB(G(P))$ and the {\sl order polytope} 
$$
O(P) := \{x \in \mathbb{R}^V : x \in [0,1]^V, x_v \leq x_w\ \text{whenever}\ v \leqslant_P w\}
$$ have the same volume. Because $O(P)$ canonically decomposes into $e(P)$ simplices, of volume $1/n!$ each, one obtains that the volume of $O(P)$, and thus $\STAB(G(P))$, is precisely $e(P)/n!$.

Now let $\{(y_{v^-},y_{v^+})\}_{v \in V}$ denote any optimal collection of intervals consistent with $P$. These intervals define another box of volume $2^{-nH(P)}$, this time contained in $O(P)$. This directly implies $nH(\bar{P}) \leq \log e(P) + \log e \cdot n$, without using Stanley's result. An elegant, elementary proof of the inequality $\log e(P) \leq nH(\bar{P})$ was given more recently by Brightwell and Tetali (\cite{BT03}, Theorem 5.2).

The second result of Kahn and Kim~\cite{KK95} is an adversarial strategy that forces any algorithm for sorting under partial information to perform a number of queries that is close to the lower bound. We give a short proof of this, namely that any algorithm can be forced to perform $\frac{1}{2} nH(\bar{P}) \geq \frac{1}{2} \log e(P)$ queries. Initially, compute an optimal collection of intervals $\{(y_{v^-},y_{v^+})\}_{v \in V}$ consistent with $P$. When faced with the query ``is $a \leqslant b$?'', answer ``yes'' if and only if $m_a \leq m_b$ in the current collection of intervals, where $m_v$ denotes the midpoint of the interval for $v \in V$. If the answer is ``yes'' (and $a \neq b$), replace the interval for $a$ by $(y_{a^-},m_a)$ and the interval for $b$ by $(m_b,y_{b^+})$. If the answer is ``no'', replace the interval for $a$ by $(m_a,y_{a^+})$ and the interval for $b$ by $(y_{b^-},m_b)$. Since $H(P) + H(\bar{P}) = \log n$, such an answer guarantees that each comparison decreases $nH(\bar{P})$ by at most $2$. Therefore, the number of comparisons performed is at least $\frac{1}{2} nH(\bar{P})$.

We now prove Theorem~\ref{th:c1}.

\begin{proof}[Proof of Theorem~\ref{th:c1}]
The proof is by induction on $n$ and, for $n$ fixed, on the number of incomparabilities in $P$. The result being true for $n = 1$, we assume $n \geq 2$. Consider an optimal vector $x \in \mathbb{R}_+^V$ and corresponding collection of open intervals $\{(y_{v^-},y_{v^+})\}_{v \in V}$.
Let $a \in V$ be such that $y_{a^+}$ is maximum.

If $a$ is comparable to all elements of $V$, then the induction hypothesis implies
$$
n H(\bar{P}) = (n-1) H(\overline{P - a}) \leq 2\log e(P - a) = 2\log e(P).
$$

Hence, we may assume that $a$ is incomparable to some element in $V$. Let $b$ be such an element
with $y_{b^+}$ maximum. Clearly, $y_{b^+} \leq y_{a^+}$.
In fact, it must be that $y_{b^+} = y_{a^+}$: Indeed, by our choice of $a$ and $b$,
we have $y_{c^+} \leq y_{b^+}$ for every $c\in V - \{a,b\}$. Thus, if $y_{b^+} < y_{a^+}$,
then one could extend to the right the interval corresponding to $b$
by an amount of $y_{a^+} - y_{b^+}$ and still
have a collection of intervals consistent with $P$. However, this new collection
defines a corresponding vector $x' \in \mathbb{R}_+^V$ such that
$$
- \frac{1}{n} \sum_{v \in V} \log x'_v
= - \frac{1}{n} \sum_{v \in V} \log x_v + \frac{1}{n} (\log x_{b} - \log x'_{b})
< - \frac{1}{n} \sum_{v \in V} \log x_v,
$$
contradicting the optimality of $x$.

Exchanging $a$ and $b$ if necessary, we may assume that $x_{a} \geq x_{b}$.
By shortening the intervals of $a$ and $b$ in two different ways,
we will define two collections of open intervals
$\{(y^{1}_{v^-},y^{1}_{v^+})\}_{v \in V}$ and $\{(y^{2}_{v^-},y^{2}_{v^+})\}_{v \in V}$.
In the first one, we will have $y^{1}_{a^+} \leq y^{1}_{b^-}$, while for the second $y^{2}_{b^+} \leq y^{2}_{a-}$ will hold. To this aim, we introduce a few quantities.

Let $\lambda := x_{b} / x_{a}$ (thus, $\lambda \in [0,1]$).
Let
$$
\alpha_{1} := \left\{
\begin{array}{ll}
\frac{1}{1-\lambda} & \textrm{ if } \lambda \leq \frac{1}{2} \\
2 & \textrm{ otherwise}
\end{array}
\right.
\qquad\qquad
\beta_{1} := \left\{
\begin{array}{ll}
1 & \textrm{ if } \lambda \leq \frac{1}{2} \\
2\lambda & \textrm{ otherwise}
\end{array}
\right.
$$
and
$$
\alpha_{2} := 2 / \lambda
\qquad\qquad
\beta_{2} := 2.
$$
The collection $\{(y^{1}_{v^-},y^{1}_{v^+})\}_{v \in V}$ equals
$\{(y_{v^-},y_{v^+})\}_{v \in V}$, except that
\begin{align*}
y^{1}_{a^+} &:= y_{a^-} + \frac{x_{a}}{\alpha_{1}} \\
y^{1}_{b^-} &:= y_{b^+} - \frac{x_{b}}{\beta_{1}}.
\end{align*}
Similarly, $\{(y^{2}_{v^-},y^{2}_{v^+})\}_{v \in V}$ equals
$\{(y_{v^-},y_{v^+})\}_{v \in V}$ with the following two exceptions:
\begin{align*}
y^{2}_{a^-} &:= y_{a^+} - \frac{x_{a}}{\alpha_{2}} \\
y^{2}_{b^+} &:= y_{b^-} + \frac{x_{b}}{\beta_{2}}.
\end{align*}

Let $P_{i}$ ($i=1,2$) be the interval order defined by $\{(y^{i}_{v^-},y^{i}_{v^+})\}_{v \in V}$, with $v \leqslant_{P_i} w$ whenever $y^{i}_{v^+} \leq y^{i}_{w^-}$.
Clearly, both $P_{1}$ and $P_{2}$ extend $P$.

We claim that there exists an index $i\in\{1,2\}$ such that
\begin{equation}
\label{eq:index}
\frac{e(P_{i})}{e(P)} \leq \frac{1}{\sqrt{\alpha_{i}\beta_{i}}}.
\end{equation}
This is proved below. Assuming that the claim is correct, let $x'\in \mathbb{R}_+^V$ be the vector defined by the collection
of open intervals $\{(y^{i}_{v^-},y^{i}_{v^+})\}_{v \in V}$.
This vector gives an upper bound on the entropy of $P_{i}$, namely
$$
H(P_{i}) \leq - \frac{1}{n} \sum_{v \in V} \log x'_v
= - \frac{1}{n} \sum_{v \in V} \log x_v + \frac{1}{n}\log\alpha_{i} + \frac{1}{n}\log\beta_{i}.
$$
Hence,
\begin{equation}
\label{eq:cost}
nH(P_{i}) \leq nH(P) + \log\left( \alpha_{i}\beta_{i} \right).
\end{equation}

Using \eqref{eq:index}, \eqref{eq:cost}, and the induction hypothesis on $P_{i}$,
we obtain
\begin{align*}
nH(\bar P) &= n\log n - nH(P) \\
&\leq n\log n - nH(P_{i}) + \log\left( \alpha_{i}\beta_{i} \right) \\
&= nH(\bar{P_{i}}) + \log\left( \alpha_{i}\beta_{i} \right) \\
&\leq 2\log e(P_{i}) + \log\left( \alpha_{i}\beta_{i} \right) \\
&\leq 2\log \frac{e(P)}{\sqrt{\alpha_{i}\beta_{i}}} + \log\left( \alpha_{i}\beta_{i} \right) \\
&= 2\log e(P).
\end{align*}

In order to prove the claim that there exists an index $i\in\{1,2\}$ such that \eqref{eq:index} holds, we show the following two inequalities:
\begin{eqnarray}
\label{eq:eP1eP2}
\frac{e(P_{1})}{e(P)} + \frac{e(P_{2})}{e(P)} &\leq &1\\
\label{eq:roots}
\frac{1}{\sqrt{\alpha_{1}\beta_{1}}} +
\frac{1}{\sqrt{\alpha_{2}\beta_{2}}} &\geq &1.
\end{eqnarray}

For proving \eqref{eq:eP1eP2}, it is enough to show
$a <_{P_{1}} b$ and $a >_{P_{2}} b$.
By definition of $\alpha_{1}$ and $\beta_{1}$, we have
$$
y^{1}_{a^+} = y_{a^-} + \frac{x_{a}}{\alpha_{1}}
= \left\{
\begin{array}{ll}
y_{a^-} + x_{a} - x_{b}  & \textrm{ if } \lambda \leq \frac{1}{2} \\
y_{a^-} + x_{a}/2 & \textrm{ otherwise}
\end{array}
\right.
$$
and (using $y_{a^+} = y_{b^+}$)
$$
y^{1}_{b^-} = y_{b^+} - \frac{x_{b}}{\beta_{1}}
= y_{a^-} + x_{a} - \frac{x_{b}}{\beta_{1}}
= \left\{
\begin{array}{ll}
y_{a^-} + x_{a} - x_{b}  & \textrm{ if } \lambda \leq \frac{1}{2} \\
y_{a^-} + x_{a}/2 & \textrm{ otherwise.}
\end{array}
\right.
$$
Hence, $a <_{P_{1}} b$. Also,
$$
y^{2}_{a^-} = y_{a^+} - \frac{x_{a}}{\alpha_{2}} = y_{a^+} - x_{b} / 2
$$
and
$$
y^{2}_{b^+} = y_{b^-} + \frac{x_{b}}{\beta_{2}}
= y_{a^+} - x_{b} + \frac{x_{b}}{\beta_{2}}
= y_{a^+} - x_{b} / 2,
$$
implying $a >_{P_{2}} b$. Therefore, \eqref{eq:eP1eP2} holds true.

We proceed and show \eqref{eq:roots}. The left-hand side of \eqref{eq:roots} is a function of $\lambda$, which we denote by $f(\lambda)$ for short. We have
$$
f(\lambda) = \left \{
\begin{array}{ll}
\displaystyle
\sqrt{1 - \lambda} + \frac{\sqrt{\lambda}}{2}  & \textrm{ if } \lambda \leq \frac{1}{2} \\[2ex]
\displaystyle
\frac{1}{2\sqrt{\lambda}} + \frac{\sqrt{\lambda}}{2}  & \textrm{ otherwise.}
\end{array}
\right.
$$
(Note that $\sqrt{1 - \lambda} + \frac{\sqrt{\lambda}}{2} =
\frac{1}{2\sqrt{\lambda}} + \frac{\sqrt{\lambda}}{2}$ for $\lambda=\frac12$.)
The first derivative of $f(\lambda)$ is
$$
f'(\lambda) = \left \{
\begin{array}{ll}
\displaystyle
\frac{1}{4\sqrt{\lambda}} - \frac{1}{2\sqrt{1 - \lambda}}  & \textrm{ if } \lambda \leq \frac{1}{2} \\[2ex]
\displaystyle
\frac{1}{4\sqrt{\lambda}} - \frac{1}{4\lambda^{3/2}}  & \textrm{ otherwise.}
\end{array}
\right.
$$
As the reader will easily check, $f'$ is positive over the open interval $(0, \frac{1}{5})$
and negative over $(\frac{1}{5}, 1)$. Since $f(0)=f(1)=1$, we deduce that
$f(\lambda) \geq 1$ for every $\lambda \in [0,1]$, as claimed.
This concludes the proof.
\end{proof}

Finally, we sketch a simple proof of the weaker inequality $nH(\bar{P}) \leq 4 \log e(P)$. We follow the same proof structure as above. Instead of picking $a \in V$ such that $y_{a^+}$ is maximum,  pick $a \in V$ such that $x_a$ is maximum. Then pick $b \in V - \{a\}$ such that the interval for $b$ contains the midpoint of the interval for $a$. If $e(P(a<b)) \leq e(P)/2$, then define $x'$ by replacing the interval for $a$ by its first quarter and the interval for $b$ by its last quarter. Otherwise, we have $e(P(b<a)) \leq e(P)/2$. In this case, define $x'$ by replacing the interval for $a$ by its last quarter and the interval for $b$ by its first quarter.

\section{Insertion Sort}
\label{sec:insertion}

We first propose an $O(n^2)$ sorting algorithm with query complexity $O(\log n \cdot \log e(P))$. It consists of first finding a maximum chain $C\subseteq P$, then iteratively inserting the remaining elements of $P - C$ in the chain $C$, using binary search, see Algorithm \ref{alg:insertion_sort}. In order to show that its query complexity is $O(\log n \cdot \log e(P))$, we need two lemmas.

\begin{algorithm}[ht]
\caption{Insertion sort-like algorithm for sorting under partial information} \label{alg:insertion_sort}
\begin{algorithmic}
\STATE \COMMENT{\textit{Phase 1 (preprocessing)}} 
\STATE find a maximum chain $C\subseteq P$ 
\STATE \COMMENT{\textit{Phase 2 (sorting)}}
\WHILE{$P - C\not= \emptyset$} 
\STATE remove an element of $P - C$ and insert it in $C$ with a binary search
\ENDWHILE
\RETURN $C$
\end{algorithmic}
\end{algorithm}

\begin{lemma}
\label{lem:bigchain}
Let $P$ be a poset of order $n$ and let $C$ be a maximum chain in $P$. Then $|C| \geq 2^{-H(\bar P)}n$.
\end{lemma}
\begin{proof}
It is well known~(\cite{K86,CFJ07}) that the entropy of a graph on $n$ vertices with stability number $\alpha$ is at least $-\log \frac{\alpha}n$.
The result follows by applying this to $\bar{G}(P)$.
\end{proof}

\begin{lemma}
\label{lem:redfac}
For all $x \in \mathbb{R}$, $1 - 2^{-x} \leq \ln 2 \cdot x$.
\end{lemma}

Now the number of comparisons performed by the algorithm is at most
\begin{eqnarray*}
\log n \cdot (n - |C|) & \leq & \log n \cdot ( n - 2^{-H(\bar P)}n ) \qquad \text{(Lemma~\ref{lem:bigchain})} \\
& \leq & \log n \cdot \ln 2 \cdot n H(\bar P) \qquad \text{(Lemma~\ref{lem:redfac})} \\
& = & O(\log n \cdot \log e(P))\qquad \text{(Lemma~\ref{lem:kk})}.
\end{eqnarray*}

This algorithm has the property that we can perform the preprocessing step only once, and sort all instances with the same partial information in time $O(\log n \cdot \log e(P)) + O(n)$. To achieve this, we store the maximum chain $C$ in a balanced binary search tree in time $O(n)$ and insert each remaining element in time $O(\log n)$.

\section{Merge Sort}
\label{sec:merge}

In order to improve on the previous algorithm, we use an approach similar to merge sort, see Algorithm \ref{alg:merge_sort}. This algorithm is illustrated in Figure~\ref{fig:merge_sort}.

\begin{algorithm}
\caption{Mergesort-like algorithm for sorting under partial information} \label{alg:merge_sort}
\begin{algorithmic}
\STATE \COMMENT{\textit{Phase 1 (preprocessing)}}
\STATE find a greedy chain decomposition $C_1$, \ldots, $C_k$ of $P$
\STATE $\mathcal{C} \gets \{C_1,\ldots,C_k\}$
\STATE \COMMENT{\textit{Phase 2 (sorting)}}
\WHILE{$|{\mathcal{C}}|>1$}
\STATE pick the two smallest chains $C$ and $C'$ in $\mathcal{C}$
\STATE merge $C$ and $C'$ into a chain $C''$, in linear time
\STATE remove $C$ and $C'$ from $\mathcal{C}$, and replace them by $C''$
\ENDWHILE
\RETURN the unique chain in $\mathcal{C}$
\end{algorithmic}
\end{algorithm}

\begin{figure}
\begin{center}
\includegraphics[scale=1]{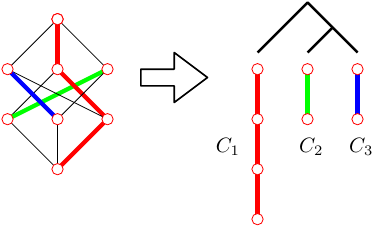}
\end{center}
\caption{\label{fig:merge_sort}Illustration of Algorithm~\ref{alg:merge_sort}.}
\end{figure}

Let $\tilde{g}$ denote the entropy of the probability distribution $\big\{\frac{|C_1|}{n},\ldots,\frac{|C_k|}{n}\big\}$, the distribution of the sizes of the chains in the greedy chain decomposition. Our next lemma bounds the query complexity of Algorithm \ref{alg:merge_sort} in terms of $\tilde{g}$.
\begin{lemma}
\label{lem:Huffman}
The query complexity of Algorithm \ref{alg:merge_sort} is at most $(\tilde{g}+1)n$.
\end{lemma}
\begin{proof}
Phase 2 of Algorithm \ref{alg:merge_sort} is a multiway merge of the chains $C_i$ extracted from $P$. The two smallest chains are iteratively merged, thereby forming a Huffman tree: this is a known strategy for merging sorted sequences of different lengths (see for instance~\cite{FB72}). Huffman codes have average codeword length within one bit of the entropy~\cite{TC}. Hence the average root-to-leaf distance in the tree with respect to the distribution $\big\{\frac{|C_1|}{n},\ldots,\frac{|C_k|}{n}\big\}$ is at most $\tilde{g} + 1$.

Merging two chains is done in linear time by iteratively choosing the minimum. Consider an element of the chain $C_i$. In the worst case, this element is compared at every node of the path from the leaf node corresponding to $C_i$, to the root of the tree. Every time we compare two elements at one node of the tree, we charge the comparison to the element that is selected. We denote by $t_i$ the length of the path to the $i$th chain $C_i$. Summing over all the elements, we get:
$$
\sum_i t_i |C_i| = n \sum_i t_i \frac{|C_i|}n \leq (\tilde{g} + 1)n,
$$
proving the lemma.
\end{proof}

The following theorem uses this bound and Theorem~\ref{thm:greedy}.
\begin{theorem}
\label{thm:nearopt}
For any $\varepsilon >0$, the query complexity of Algorithm \ref{alg:merge_sort} is at most
$$
(1+\varepsilon) \log e(P) + (1+\varepsilon) \Big(\log e + \log \Big(1+\frac{1}{\varepsilon}\Big)\Big)\,n + n = 
(1+\varepsilon) \log e(P) + O_\varepsilon(n).
$$
\end{theorem}

\begin{proof}
From Lemma \ref{lem:Huffman}, we infer that the query complexity is at most
\begin{eqnarray*}
(\tilde{g} + 1)\,n
& \leq & (1+\varepsilon) nH(\bar{P}) + (1+\varepsilon)\,\log \Big(1+\frac{1}{\varepsilon}\Big)\,n + n \qquad \text{(Theorem~\ref{thm:greedy})}\\
& \leq & (1+\varepsilon)\,(\log e(P) + \log e \cdot n) + (1+\varepsilon)\,\log \Big(1+\frac{1}{\varepsilon}\Big)\,n + n\qquad \text{(Lemma \ref{lem:kk})}.
\end{eqnarray*}
\end{proof}
We conclude that Algorithm \ref{alg:merge_sort} is an algorithm with query complexity at most $(1+\varepsilon)\log e(P) + O_\varepsilon(n)$, for any $\varepsilon > 0$. It is shown in Appendix~\ref{app:greedy} that the greedy chain decomposition can be performed in time $O(n^{2.5})$. This is actually the bottleneck of the algorithm, and the global complexity of Algorithm \ref{alg:merge_sort} is $O(n^{2.5})$ as well. Hence any improvement on the greedy chain decomposition algorithm would yield an improved sorting algorithm.

Note that again, we can reuse the chain decomposition obtained in the preprocessing phase for sorting any instance with the same partial information in time proportional to the query complexity.

\section{Cautious Merge Sort}
\label{sec:cautious_merge}

The query complexity of Algorithm \ref{alg:merge_sort} is not $O(\log e(P))$ because it completely ignores a large part of the partial information. Now, we show that using the partial information for the {\sl last\/} merge suffices to obtain an algorithm with query complexity $O(\log e(P))$.

The subproblem at hand is that of {\sl merging under partial information\/}. It is a special case of sorting under partial information, in which the given poset $P$ is covered by two chains $A$ and $B$, that is, $P$ is of width at most $2$. (The {\sl width\/} of a poset $P$ is the maximum size of an antichain of $P$.)

That problem was studied by Linial~\cite{L84}, who proposed an algorithm with query complexity $O(\log e(P))$. However, this algorithm requires computing polynomially many $\Theta (n)\times \Theta (n)$ determinants. In Section \ref{sec:MUPIsec}, we obtain an $O(n^2 \log^2 n)$ algorithm for the problem with query complexity at most $6 \log e(P)$. 

\begin{theorem}
\label{thm:mupired}
Suppose there exists an algorithm for the problem of merging under partial information with query complexity at most $c_3 \log e(P')$, given as partial information a poset $P'$ of order $n$ and width at most $2$. Then there exists an algorithm for the problem of sorting under partial information with query complexity at most $(9.09 + c_3) \log e(P)$.
\end{theorem}
\begin{proof}
Let Algorithm \ref{alg:MUPI} be the hypothesized algorithm for merging under partial information. 
(Such an algorithm will be given in Section \ref{sec:MUPIsec}.) 
Consider the following algorithm, illustrated in Figure~\ref{fig:mupireduction}.

\begin{algorithm}[h!]
\caption{Improved mergesort-like algorithm for sorting under partial information} \label{alg:cautious_merge_sort}
\begin{algorithmic}[1]
\STATE find a maximum chain $A \subseteq P$
\STATE \label{line:merge_B} apply Algorithm \ref{alg:merge_sort} to the poset $P - A$, yielding a chain $B$
\STATE \label{line:call_MUPI} apply Algorithm \ref{alg:MUPI} to the current partial information $P'$
\RETURN the resulting chain
\end{algorithmic}
\end{algorithm}

\begin{figure}
\begin{center}
\includegraphics[scale=1]{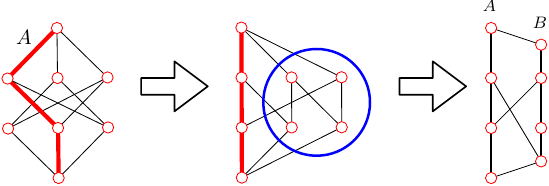}
\end{center}
\caption{\label{fig:mupireduction}Illustration of Algorithm~\ref{alg:cautious_merge_sort}.}
\end{figure}

From Lemma~\ref{lem:bigchain}, we have $|A|\geq 2^{-H(\bar P)} n$, and therefore (using Lemma~\ref{lem:redfac}):
\begin{equation}
\label{eq:rest}
|B| = |P-A| \leq n (1-2^{-H(\bar P)}) \leq \ln 2 \cdot nH(\bar P) .
\end{equation}
Now from Theorem~\ref{thm:nearopt} the number of comparisons in lines \ref{line:merge_B} and \ref{line:call_MUPI} of Algorithm \ref{alg:cautious_merge_sort} is at most
$$
\begin{array}{cl}
& (1+\varepsilon) \log e(P-A) + \Big( (1+\varepsilon)\big(\log e + \log (1 + \frac{1}{\varepsilon})\big) + 1 \Big) |P-A| + c_3 \log e(P')\\[2ex]
\leq & (1+\varepsilon) \log e(P) + \Big( (1+\varepsilon)\big(1 + \ln (1 + \frac{1}{\varepsilon})\big) + \ln 2 \Big)\,nH(\bar P) + c_3 \log e(P) \qquad \text{(from (\ref{eq:rest}))}\\[2ex]
\leq & \Big( 1 + \varepsilon + 2 \Big( (1+\varepsilon)\big(1 + \ln (1 + \frac{1}{\varepsilon})\big) + \ln 2 \Big) + c_3\Big) \log e(P) \qquad \text{(Theorem~\ref{th:c1})}\\[2ex]
\leq & (9.09 + c_3) \log e(P) \qquad \text{(taking $\varepsilon = 0.35$)}.
\end{array}
$$
\end{proof}

Assuming that the global complexity of Algorithm \ref{alg:MUPI} is $O(T(n))$, the results of the previous section imply that the global complexity of Algorithm \ref{alg:cautious_merge_sort} is $O(n^{2.5}) + O(T(n))$.

\section{Merging under Partial Information}
\label{sec:MUPIsec}

In this section, we assume that $P$ is covered by two disjoint chains, denoted by $A$ and $B$. First, we describe a structural result by K\"orner and Marton~\cite{KM88} concerning the entropy of a bipartite graph. Second, we show how to use this to obtain an $O(n^2 \log^2 n)$ algorithm with query complexity at most $c_3 \log e(P)$, with $c_3 = 6$.

Before proceeding, we state a lemma providing several properties of the incomparability graph of $P$ that we repeatedly use subsequently. The proof is straightforward, thus omitted.

\begin{lemma}
\label{lem:width_2_prop}
Let $P$ be a poset covered by two disjoint chains $A$ and $B$, and let $G = \bar{G}(P)$. Then:
\begin{enumerate}[(i)]
\item The graph $G$ is bipartite, with bipartition $A$, $B$;
\item The neighborhood $N(u)$ of any vertex $u$ in $G$ is an interval in the opposite chain (thus, $G$ is {\em biconvex});
\item Consider two vertices $u$ and $v$ in the same chain, say $A$, and such that $u \leqslant_P v$. Let $[c_u,d_u]$ and $[c_v,d_v]$ denote the intervals of $B$ defined by $N(u)$ and $N(v)$, respectively. Then, we have $c_u \leqslant_P c_v$ and $d_u \leqslant_P d_v$. In particular, $N(w)$ is contained in the interval $[c_u,d_v]$ of $B$ with endpoints $c_u$ and $d_v$, whenever $w$ belongs to $A$ and $u \leqslant_P w \leqslant_P v$.
\end{enumerate}
\end{lemma}

\subsection{The Entropy of Bipartite Graphs}
\label{sec:H_bip}

As noted above in Lemma \ref{lem:width_2_prop}(i), the incomparability graph $\bar{G}(P)$ of $P$ is a bipartite graph. K\"orner and Marton~\cite{KM88} describe a method for computing the entropy of {\sl any\/} bipartite graph (see Theorem 3.8 in Simonyi's survey on graph entropy~\cite{S95}). Below, $h$ denotes the binary entropy function. Thus, we have $h(\xi) = -\xi \log \xi - (1-\xi) \log (1-\xi)$ for $\xi \in (0,1)$, and $h(0) = h(1) = 0$.

\begin{theorem}[\cite{KM88}]
\label{thm:KM}
Let $G$ be a bipartite graph of order $n$, with bipartition $A$, $B$. Then, one can find partitions $A = A_1 \cup \cdots \cup A_k$ and $B = B_1 \cup \cdots \cup B_k$ such that
$$
H(G) = \sum_{i=1}^k \frac{|A_i| + |B_i|}{n} \, h\left(\frac{|A_i|}{|A_i| + |B_i|}\right).
$$
\end{theorem}

The partitions are constructed iteratively. Let $N_G(X)$ denote the neighborhood of a set $X$ of vertices in the graph $G$. For $i \in \{1,\ldots,k\}$, K\"orner and Marton define $A_i$ as any subset of $A' := A - A_1 - \cdots - A_{i-1}$ that maximizes
\begin{equation}
\label{eq:bip_ratio}
\frac{|A_i|}{|N_{G'}(A_i)|}
\end{equation}
in the graph $G'$ obtained from $G$ by removing all vertices contained in some $A_j$ or some $B_j$ with $j < i$, and define $B_i$ as $N_{G'}(A_i)$. By convention, if there is a vertex $u$ in $A'$ that is isolated, then we let $A_i = \{u\}$ and $B_i = \varnothing$. If $A'$ is empty and $B' := B - B_1 - \cdots - B_{i-1}$ is not, we pick a vertex $v$ in $B'$, let $A_i = \varnothing$ and $B_i = \{v\}$.

Given partitions as in Theorem \ref{thm:KM}, the point $x \in \STAB(G)$ achieving the minimum in (\ref{def-H}) is obtained by simply letting
$$
x_u := \frac{|A_i|}{|A_i| + |B_i|}\ \ \text{whenever}\ \ u \in A_{i},
\qquad \text{and} \qquad
x_v := \frac{|B_i|}{|A_i| + |B_i|}\ \ \text{whenever}\ \ v \in B_{i}.
$$

\subsection{Local Optimality and rebalancing}
\label{sec:locimp}

Let $G = \bar{G}(P)$, and let $E$ denote the edge set of $G$. Because $G$ is bipartite, 
$$
\STAB(G) = \{x \in \mathbb{R}^V : x_u + x_v \leq 1\ \forall uv \in E,\ 0 \leq x_v \leq 1\ \forall v \in V\}.
$$
Consider a point $x$ in $\STAB(G)$. The point $x$ is a feasible solution of the convex program \eqref{def-H}. If $x_v = 0$ for some vertex $v \in V$, then the objective function value is infinite, and the point $x$ is useless. In order to prevent this, we mostly consider points in $\STAB^*(G) := \STAB(G) \cap \{x \in \mathbb{R}^V : x_v > 0\ \forall v \in V\}$. 

An edge $uv \in E$ is said to be {\sl tight\/} with respect to the solution $x$ if $x_{u} + x_{v}=1$. Let $G(x)$ denote the graph whose vertices are those of $G$ and whose edges are the edges of $G$ that are tight.

We begin with a lemma that governs much of the structure of $G(x)$, about edges that `cross'. Recall that $P$ is covered by two disjoint chains $A$ and $B$. We say that two edges $uv$ and $u'v'$ of $G$, with $u, u' \in A$ and $v, v' \in B$, {\sl cross\/} if $u <_P u'$ and $v' <_P v$, or $u' <_P u$ and $v <_P v'$. 

\begin{lemma}
\label{lem:crossing}
Let $P$ be a poset covered by two disjoint chains $A$, $B$, and let $G = \bar{G}(P)$. Consider a point $x \in \STAB(G)$ and two edges $uv, u'v' \in E(G)$ that are tight with respect to $x$, with $u, u' \in A$ and $v, v' \in B$. If $uv$ and $u'v'$ cross, then both $uv'$ and $u'v$ are edges of $G$, and both are tight with respect to $x$.
\end{lemma}

\begin{proof}
By Lemma \ref{lem:width_2_prop}(iii), $uv'$ and $u'v$ belong to $E(G)$. Assume, by contradiction, that $uv'$ is not tight. Then
\begin{eqnarray*}
x_{v} &=& 1-x_{u} \quad \text{(because $uv$ is tight)}\\
    &>& x_{v'} \quad \text{(because $uv'$ is not tight)}\\
    &=& 1-x_{u'} \quad \text{(because $u'v'$ is tight)}\\
    &\geq& x_{v} \quad \text{(because $u'v$ is an edge and $x$ is feasible),}
\end{eqnarray*}
a contradiction. We conclude that both $uv'$ and $u'v$ are tight.
\end{proof}

The point $x \in \STAB(G)$ is called {\sl locally optimal\/} if, for every (connected) component $K$ of $G(x)$:
\begin{equation}
\label{K:weight}
x_{u} = \frac{|A \cap K|}{|K|}\ \ \text{for all}\ \ u \in A \cap K, 
\qquad\text{and}\qquad
x_{v} = \frac{|B \cap K|}{|K|}\ \ \text{for all} \ \ v \in B \cap K.
\end{equation}
We say that the component $K$ is {\sl balanced\/} if the local optimality condition \eqref{K:weight} holds. Otherwise $K$ is {\sl unbalanced}. 

Consider a point $x \in \STAB^*(G)$. A component of $G(x)$ is {\sl trivial\/} if it consists of a unique vertex, {\sl non-trivial\/} otherwise. A trivial component can be either balanced or unbalanced, in which case it is said to be {\sl loose}. Observe that a trivial component $\{v\}$ can be balanced (that is, $x_v=1$) only if $v$ is a cut-point of $P$, that is, $v$ is comparable to every other vertex of $P$. (Here we use the assumption $x_u > 0$ for the vertices $u \in N_G(v)$.)

The first part of the next lemma states that a component $K$ of $G(x)$ typically determines two (possibly trivial, or even empty) intervals, one in the chain $A$ and the other in the chain $B$. The exceptions are characterized by the following definition: we say that a component $L$ of $G(x)$ is an {\sl inlay\/} of another component $K$ if there exists a vertex $w \in L$ and vertices $u, u'' \in K$ in the same chain as $w$ (that is, $u, u'' \in A$ iff $w \in A$) such that $u \leqslant_P w \leqslant_P u''$. 

The second part of the lemma implies that $P$ induces a linear order on the non-trivial components of $G(x)$. This linear order naturally extends to all components of $G(x)$, provided that no such component is loose. 

Below, when $S$ and $T$ are two disjoint subsets of the poset $P$, we write $S \leqslant_P T$ whenever $u \leqslant_P v$ holds for every $u \in S$ and every $v \in T$.

\begin{lemma}
\label{lem:int}
Let $P$ be a poset covered by two disjoint chains $A$, $B$, and let $G = \bar{G}(P)$. Consider a point $x$ in $\STAB^*(G)$. Then
\begin{enumerate}[(i)]
\item if a component $L$ of $G(x)$ is an inlay of a component $K$ of $G(x)$, then $L$ is trivial and loose;
\item if $K$, $L$ are distinct non-trivial components of $G(x)$, then either $K \leqslant_P L$ or $L \leqslant_P K$.
\end{enumerate}
\end{lemma}
\begin{proof}
(i) Suppose otherwise. Let $w \in L$ and $u, u'' \in K$ be as above. Without loss of generality, we may assume that all three vertices belong to $A$ and $u' \notin K$ whenever $u' \in A$ and $u <_P u' <_P u''$. By Lemma \ref{lem:width_2_prop}(iii), because $K$ is a component of $G(x)$ containing $u$ and $u''$, there is a vertex $v$ of $K \cap B$ adjacent to both $u$ and $u''$ in $G(x)$.

By Lemma \ref{lem:width_2_prop}(ii), the neighborhood of $v$ in $G$ is an interval in $A$ containing $u$ and $u''$. Thus, it also contains $w$. Because $w$ does not belong to $K$, the edge $vw$ is not tight with respect to $x$. 

First, suppose that $L$ is non-trivial. Thus there exists $w' \in B, w' \neq v$ such that $ww' \in E(G)$ and the edge $ww'$ is tight with respect to $x$. However, $ww'$ crosses either $uv$ or $u''v$, implying in both cases that $vw$ is tight by Lemma \ref{lem:crossing}, a contradiction. Hence, $L$ is trivial and $L = \{w\}$.

Second, suppose that $L$ is balanced. Because $x \in \STAB^*(G)$, we have $x_w + x_v = 1 + x_v > 1$, a contradiction. Hence, $L$ is loose.

(ii) On the contrary, suppose that neither $K \leqslant_P L$ nor $L \leqslant_P K$ holds. From what precedes, neither $K$ nor $L$ is an inlay. Consequently, we may assume $K \cap A \leqslant_P L \cap A$ and at the same time $L \cap B \leqslant_P K \cap B$, without loss of generality. Let $uv$ and $u'v'$ be edges of the components $K$ and $L$, respectively, with $u,u' \in A$ and $v,v' \in B$. By our assumption, these two edges cross. Hence, by Lemma \ref{lem:crossing} the edge $uv'$ is tight with respect to $x$. This implies that $u$ and $v'$ are in the same component of $G(x)$, a contradiction. 
\end{proof}
  
Our algorithm for merging under partial information will take in input a locally optimal point $x \in \STAB(G)$. It will repeatedly modify $G$ and $x$, and then ``rebalance'' $x$ so that it becomes locally optimal again.

The rebalancing algorithm is described in Algorithm~\ref{alg:rebalance}. Its input is a point $x \in \STAB^*(G)$. Given a component $K$ of $G(x)$, the {\sl slack} %% find a better name??
of $K$ is defined as the real $\sigma$ minimizing 
$$\left|x_{v} + \sigma - \frac{|A \cap K|}{|K|}\right|$$
under the constraint that $x' \in \STAB^*(G)$, where $v$ is any vertex in $A\cap K$ and $x'$ is obtained from $x$ by adding $\sigma$ to $x_{w}$ for all $w\in A\cap K$, and substracting $\sigma$ from $x_{w}$ for all $w\in B\cap K$. In other words, $\sigma$ represents the maximum quantity by which we can ``rebalance'' $K$ without losing feasibility. (Note that the slack could be negative.)

\begin{algorithm}[h!t]
\caption{rebalancing Algorithm} \label{alg:rebalance}
\begin{algorithmic}[1]
\WHILE{there exists an unbalanced component $K$ in $G(x)$}
\STATE \label{line:compute_slack} compute the slack $\sigma$ of $K$
\STATE \label{line:add_slack} put $x_{v} := x_{v} + \sigma$ for all $v \in A \cap K$, and $x_{v} := x_{v} - \sigma$ for all $v \in B \cap K$
\ENDWHILE
\end{algorithmic}
\end{algorithm}

Here are a few properties of Algorithm~\ref{alg:rebalance} which are easy to check. Consider an iteration of the while-loop, and let $x'$ be the modified point $x$ at the end of the iteration. 

When the component $K$ is ``rebalanced'', either it becomes a balanced component of the graph $G(x')$, or there is at least one edge in $G$ between $K$ and another component of $G(x)$ that became tight with respect to $x'$, and hence $K$ is ``merged'' with other components of $G(x)$ into a single component of $G(x')$. 

In order to capture when this happens, we say that $K$ {\sl touches\/} another component $L$ of $G(x)$ if $K$ and $L$ are linked by an edge of $G$ (that is not tight with respect to $x$) and there exists no non-trivial component $M$ distinct from $K$ and $L$ such that some edge $yz$ with both endpoints in $M$ is ranked between $K$ and $L$, that is, $K \cap A \leqslant_P \{y\} \leqslant_P L \cap A$ and $K \cap B \leqslant_P \{z\} \leqslant_P L \cap B$, or $L \cap A \leqslant_P \{y\} \leqslant_P K \cap A$ and $L \cap B \leqslant_P \{z\} \leqslant_P K \cap B$ (we assume $y \in A$ and $z \in B$). 

It follows from Lemma \ref{lem:int}(ii) that $K$ touches at most two non-trivial components of $G(x)$. 

\begin{lemma}
\label{lem:rebalancing_touch}
Let $x \in \STAB^*(G)$ and $K$ be as above. If $K$ merges with a component $L$ of $G(x)$ then it touches $L$.
\end{lemma}

\begin{proof}
Consider any edge $vw$ that caused $K$ and $L$ to merge, with $v \in K$ and $w \in L$. If $K$ and $L$ do not touch, there exists a component $M$ and an edge $yz$ as above. By Lemma~\ref{lem:crossing} applied to $G$ and $x'$, where $x' \in \STAB^*(G)$ is defined as precedingly, we conclude that $M$ and $L$ are contained in the same component of $G(x')$. Because $x'_u = x_u$ for all vertices outside $K$, this implies that $M$ and $L$ are contained in the same component of $G(x)$, a contradiction.
\end{proof}

Considering a point $\tilde x \in \STAB^*(G)$, we color the components of $G(\tilde x)$ as follows: a component is colored {\sl red} if it has at least as many vertices in $A$ than in $B$, {\sl blue} otherwise. The point $\tilde x$ is said to be {\sl color consistent}
if for every component $\tilde K$ of $G(\tilde x)$, and vertices $u \in A \cap \tilde K$ and $v \in B \cap \tilde K$, we have $\tilde x_{u} \geq 1/2$ and $ \tilde x_{v} \leq 1/2$ if $\tilde K$ is red, and $\tilde x_{u} < 1/2$ and $\tilde x_{v} > 1/2$ if $\tilde K$ is blue. Observe that being color consistent is a relaxation of being locally optimal. 

\begin{lemma}
\label{lem:rebalancing_colors}
Let $x, x' \in \STAB^*(G)$ and $K$ be as above.
If $x$ is color consistent, then $x'$ is also color consistent. Moreover, $K$ cannot merge with components that have colors different from that of $K$, and the component of $G(x')$ containing $K$ has the same color as $K$.
\end{lemma}

\begin{proof}
First observe that, because $x$ is color consistent, the weight modification is such that, for $v \in A \cap K$, we have $x'_{v} \geq 1/2$ iff $x_{v} \geq 1/2$, and similarly, for $v \in B \cap K$, we have $x'_{v} > 1/2$ iff $x_{v} > 1/2$. Thus $x'$ is also color consistent if $K$ does not merge with other components.
Now assume that $K$ does merge with other components.
Consider an edge $vw$ of $G$ between $K$ and another component $L$ of $G(x)$ that became tight with respect to $x'$, with $v \in K$ and $w\in L$. Since $x'_{w} = x_{w}$ and $x_{v} + x_{w} < x'_{v} + x'_{w} = 1$, we have $x'_{v} > x_{v}$. There are four cases to consider:
\begin{itemize}
\item\ $v \in A$ and $K$ is red in $G(x)$. Then $x'_{v} > x_{v} \geq 1/2$ and $x_{w} = x'_{w} < 1/2$. 
\item\ $v \in A$ and $K$ is blue in $G(x)$. Then $x_{v} < 1/2$, hence $x'_{v} < 1/2$ and $x_{w} = x'_{w} > 1/2$. 
\item\ $v \in B$ and $K$ is red in $G(x)$. Then $x_{v} \leq 1/2$, hence $x'_{v} \leq 1/2$ and $x_{w} = x'_{w} \geq 1/2$. 
\item\ $v \in B$ and $K$ is blue in $G(x)$. Then $x'_{v} > x_{v} > 1/2$ and $x'_{w} = x_{w} < 1/2$.  
\end{itemize}
In each case we conclude that $L$ had the same color as $K$ in $G(x)$, as claimed. Considering all components $L$ that merge with $K$, a direct consequence of what precedes is that the component of $G(x')$ containing $K$ has also the same color as $K$. Furthermore, $x'$ is color consistent.
\end{proof}

Finally, observe that the entropy of $x'$ is at most that of $x$. This is because the function
$$
\xi \mapsto \frac{|A \cap K|}{|K|} \log \xi +
\frac{|B \cap K|}{|K|} \log (1-\xi)
$$
is strictly convex over the interval $(0,1)$ with a minimum in $\xi = |A \cap K| / |K|$.

\subsection{The Core of the Algorithm}
\label{sec:MUPI_core}

At the heart of our algorithm for merging under partial information is the following procedure. Given a locally optimal point $x \in \STAB(G)$, carefully pick a non-trivial component of $G(x)$ and merge the two corresponding chains. This causes all the edges between the two chains to disappear from $G$. This also creates loose components. Then, make $x$ color consistent by increasing certain coordinates of $x$ to $1/2$ or $1/2+\varepsilon$. Finally, make $x$ locally optimal again by rebalancing it. At all times, the point $x$ remains feasible, that is, $x\in \STAB(G)$. This procedure is repeated as long as it is necessary. In the process, some vertices become cut-points, which reveals their respective ranks, and are copied in an output chain.

This time, for merging a pair of chains, we use the Hwang-Lin algorithm~\cite{HL72}. This is a simple near optimal algorithm for merging two disjoint chains $X$ and $Y$ of different lengths. It proceeds by splitting the longest chain, say $X$, into blocks of size $2^{\lfloor \log (|X|/|Y|)\rfloor}$. Then every vertex in the smallest chain $Y$ is inserted into $X$, by first performing a linear search among the blocks, then a bisection within a block. The vertices in $Y$ are inserted in order, so that once a block of $X$ is discarded, it is never looked at again.

In the analysis of Algorithm \ref{alg:MUPI}, we will use the following bound on the number of comparisons performed by the Hwang-Lin algorithm.

\begin{lemma}
\label{lem:HL_bound}
Provided $|X| \geq |Y|$, the number of comparisons performed by the Hwang-Lin algorithm for merging two disjoint chains $X$ and $Y$ is at most $|Y| \log (4|X|/|Y|)$.
\end{lemma}
\begin{proof}
It is known \cite{HL72} that the number of comparisons performed by the Hwang-Lin merging algorithm is at most
$$
|Y| \left(1 + \left\lfloor \log \frac{|X|}{|Y|} \right\rfloor\right) + \left\lfloor \frac{|X|}{2^{\lfloor \log \frac{|X|}{|Y|} \rfloor}} \right\rfloor - 1.
$$
Let $\xi \in [0,1)$ be such that
$$
\left\lfloor \log \frac{|X|}{|Y|} \right\rfloor = \log \frac{|X|}{|Y|} - \xi.
$$
Then the number of comparisons is at most
$$
|Y| \Big(1-\xi+2^\xi + \log \frac{|X|}{|Y|}\Big) \leq |Y| \log\frac{4|X|}{|Y|},
$$
where the last inequality follows from $1-\xi+2^\xi \leq 2$ for $\xi \in [0,1)$. The result follows.
\end{proof}

Consider a locally optimal point $x \in \STAB(G)$, and a non-trivial component $K$ of $G(x)$. By Lemma \ref{lem:int}(i), $K$ consists of two disjoint chains, namely $A \cap K$ and $B \cap K$, which form intervals in the  chains $A$ and $B$, respectively. 

We define the {\sl small chain\/} of $K$ to be $A \cap K$ if $K$ is blue and $B \cap K$ if $K$ is red. The {\sl big chain} of $K$ is the other one. Thus the small chain of $K$ is the one that has minimum cardinality, except that in case the two chains $A \cap K$ and $B \cap K$ have the same cardinality then $K$ is red by definition and the small chain is $B \cap K$. Because $x$ is color consistent, all vertices $v$ in the small chain of $K$ have $x_v \leq 1/2$ (and even $x_v < 1/2$ when $K$ is blue). This is why vertices of the small chain of $K$ are called {\sl small vertices}. Similarly, the vertices in the big chain of $K$ are called {\sl big vertices}.

The component $K$ is said to be {\sl good} if all the edges of $G$ having one endpoint in its small chain have their other endpoint either in the other chain of $K$ or in a component whose color is distinct from that of $K$.

\begin{lemma}
\label{lem:good_component}
Suppose $x \in \STAB(G)$ is locally optimal. If $G(x)$ has at least one non-trivial red component, then one of them is good. The same is true for non-trivial blue components.
\end{lemma}
\begin{proof}
Suppose $G(x)$ has at least one non-trivial red component.
Let $K$ be such a component which minimizes $|A \cap K|/|K|$.
We will show that $K$ is a good component.
Consider a vertex $v\in B \cap K$ and suppose 
$w$ is a neighbor of $v$ in $G$  which is outside $K$.
Since $vw$ is not tight, $x_v + x_w < 1$. In particular, $x_w < 1$, and hence
$w$ belongs to another non-trivial component $L$ of $G(x)$. 
If $L$ is red, by our choice of $K$ we have 
$|A \cap L|/ |L| \geq |A \cap K|/|K|$. Thus
$$
x_v + x_w = \frac{|B \cap K|}{|K|} + \frac{|A \cap L|}{|L|}
\geq \frac{|B \cap K|}{|K|} + \frac{|A \cap K|}{|K|} = 1,
$$
implying that $vw$ is tight, a contradiction. 
Hence $L$ must be blue, as claimed.

The case of blue components is handled similarly.
\end{proof}

We are now ready to formally state the algorithm. For the sake of simplicity, the algorithm makes four assumptions. First, the given point $x \in \STAB(G)$ is locally optimal. Second, the contribution of the red components to the entropy of $x$ does not exceed that of the blue components. (The second assumption can be made without loss of generality: if this is not the case, simply exchange the chains $A$ and $B$.) Third, the constant $\varepsilon$ on line \ref{line:make_color_consistent} of the algorithm is set to $\varepsilon := \frac{1}{2n}$. Last, all the cut-points that $P$ initially has have already been copied to the output chain at their respective final positions.

\begin{algorithm}[h!]
\caption{Core of the Algorithm for Merging under Partial Information} \label{alg:MUPI_core}
\begin{algorithmic}[1]
\WHILE{$G(x)$ has a non-trivial component}
\STATE \label{line:pick} pick a good component $K$, giving higher priority to red components
\STATE \label{line:HL} merge the chains $X := A \cap K$ and $Y := B \cap K$ with the Hwang-Lin algorithm 
\FOR{$v\in K$}\label{line:begin_for-loop1}
\STATE \label{line:make_color_consistent} put $x_{v} := \max \{x_v,1/2\}$ if $v \in A$, and $x_v := \max \{x_v, 1/2 + \varepsilon\}$ if $v \in B$
\ENDFOR\label{line:end_for-loop1}
\STATE \label{line:rebalance_core} rebalance $x$ using Algorithm~\ref{alg:rebalance}
\FOR{$v\in K$}\label{line:begin_for-loop2}
\IF{$x_{v}=1$}
\STATE copy $v$ at its final position in the chain $C$
\ENDIF
\ENDFOR\label{line:end_for-loop2}
\ENDWHILE
\RETURN $C$
\end{algorithmic}
\end{algorithm}

When the chains $X$ and $Y$ are merged (see line \ref{line:HL} of Algorithm \ref{alg:MUPI_core}), all the edges of $G$ between vertices of $K$ disappear (in other words, the vertices in $K$ become comparable elements of $P$). As a result, every vertex of $K$ forms a loose trivial component. We will prove that increasing the corresponding coordinates of $x$ to at least $1/2$ or $1/2 + \varepsilon$ (see the for-loop in lines \ref{line:begin_for-loop1}--\ref{line:end_for-loop1} of Algorithm \ref{alg:MUPI_core}) makes $x$ color consistent, so that Algorithm \ref{alg:rebalance} can be applied.

Before proving this last fact, we now study in more detail the evolution of the graph during an iteration of the algorithm. Again, we use different symbols to denote the current objects (graph, point) at different moments of the algorithm: $G$ and $x$ respectively denote the graph and point at beginning of an iteration of the main loop (line \ref{line:pick}), $G'$ denotes the graph after the merging (lines \ref{line:begin_for-loop1}--\ref{line:end_for-loop2}), and $x'$ denotes the point just after the first for-loop (line \ref{line:rebalance_core}).

As the reader can verify, $G'$ is a spanning subgraph of $G$ containing none of the edges of $G$ with both endpoints in $K$ and some of the edges of $G$ with exactly one endpoint in $K$.

%The other edges of $G$, that is, the edges of $G$ with both endpoints outside $K$, survive in $G'$ if and only
%if both of their endpoints are ranked below $K$, or both are ranked above $K$. In particular, all the edges
%linking two vertices of the same component of $G(x)$ survive in $G'(x)$, provided this component is distinct %from $K$.

\begin{lemma}
\label{lem:MUPI_colors}
Using the notations above, $x'$ is a color consistent point of $\STAB(G')$.
\end{lemma}

\begin{proof}
Observe that $\{v\}$ is a loose component of $G'(x)$ for every $v\in K$. We will show that this remains true in the graph $G'(x')$.

Let $v\in K$. First suppose $K$ is red in $G(x)$. If $v\in A$, then $x_{v} \geq 1/2$, thus
$x'_{v}=x_{v}$, and hence $\{v\}$ remains loose in $G'(x')$. If $v\in B$, then $x_{v} \leq 1/2$,
implying $x'_{v} = 1/2 + \varepsilon$. However, since $K$ was a good component, 
every neighbor $w$ of $v$ in $G'$ belonged to a blue component of $G(x)$, and 
thus $x_{w} < 1/2$. Since $x$ is locally optimal, this implies 
$x_{w} \leq 1/2 - 1/n = 1/2 - 2\varepsilon$.
It follows
$$
x'_{v} + x'_{w} = (1/2 +\varepsilon) + x_{w} \leq (1/2 +\varepsilon) + (1/2 - 2\varepsilon) < 1,
$$
%,
implying that $vw$ is not tight w.r.t.\ $x'$. Therefore, $\{v\}$ is loose in $G'(x')$.

Now assume $K$ is blue in $G(x)$. If $v\in B$, then $x_{v} > 1/2$, and thus 
$x_{v} \geq 1/2 + 1/n = x_{v} \geq 1/2 + 2\varepsilon$ (because $x$ is locally optimal). Hence, 
$x'_{v} = x_{v}$, and the component $\{v\}$ is loose in $G'(x')$.
If $v\in A$, then $x_{v} < 1/2$ and thus $x'_{v} = 1/2$.
Since the algorithm gives priority to good components that are red, Lemma~\ref{lem:good_component}
implies that all red components in $G(x)$ are trivial. 
Since $x$ is locally optimal, none of them is loose.
Since $x_{v} > 0$ and since $K$ was a good component, it follows that $v$ is only adjacent to 
vertices in $B\cap K$ in $G$. Therefore, $v$ is an isolated vertex of $G'$, and the component
$\{v\}$ is loose in $G'(x')$.  

It follows that the components of $G'(x)$ are exactly those of $G(x)$ that are distinct from $K$ 
plus the loose components $\{v\}$ for all $v\in K$. Since 
$x'_{v} \geq 1/2$ if $v \in A \cap K$ and $x'_{v} > 1/2$ if $v \in B \cap K$, and 
because $x'_{w} = x_{w}$ for every $w \in V(G) - K$, we deduce that $x'$ is color consistent. 
\end{proof}

Observe that after the first for-loop (lines \ref{line:begin_for-loop1}--\ref{line:end_for-loop1}), all small vertices of $K$ become big, and all big vertices stay big. By Lemmas \ref{lem:rebalancing_colors} and \ref{lem:MUPI_colors}, the rebalancing step (line \ref{line:rebalance_core}) preserves the status of all the vertices, that is, small vertices stay small and big vertices stay big.

\begin{lemma}
\label{lem:MUPIqc}
Let $P$ be a poset of order $n$ covered by two disjoint chains $A$, $B$, and let $G = \bar{G}(P)$. Assume that $x \in STAB(G)$ is a locally optimal point such that the contribution of the red components to $\ent(x)$ is not larger than that of the blue components. Then, Algorithm~\ref{alg:MUPI_core} merges $A$ and $B$ in at most $3n\ent(x)$ comparisons.  
\end{lemma}
\begin{proof}
Let $k$ be the total number of iterations of the while-loop.
Consider the $j$th iteration. 
Let $G_{j}$ denote the graph $G$ at the beginning of that iteration.
In this proof, we deviate from the notations used above, and denote by $x'$, $x''$, and $x'''$ 
the feasible point under consideration at the beginning,
at the end of the first for-loop, and at the end of the while-loop, respectively. 
(We keep the notation $x$ for the original point given in input.)
Let $K$ be the good component chosen at that iteration.
Let $s_{j}$ be the number of small vertices in $K$. (Thus, $s_{j}= |K\cap A|$
if $|K\cap A| < |K\cap B|$, and $s_{j}= |K\cap B|$ otherwise.)
Let similarly $t_{j}$ be the number of big vertices in $K$.

Let $r_{j}$ be the number of small red vertices in $G_{j}(x')$,
and let $\phi_{j}:= n\ent(x') + r_{j}$. Let also $r_{k+1} := 0$ and $\phi_{k+1} := 0$.

From Lemma~\ref{lem:HL_bound}, we know that when the two chains $K\cap A$ and $K \cap B$ are merged,
the Hwang-Lin algorithm spends at most $s_{j} \log \frac{4t_{j}}{s_{j}}$ comparisons. 

First suppose $K$ is red in $G_{j}(x')$. During the first for-loop (lines \ref{line:begin_for-loop1}--\ref{line:end_for-loop1}), the algorithm increases $x'_{v}$ to at least $1/2$ for every small vertex $v$ in $K$. Thus, while such a vertex contributed $-\frac 1n \log \frac {s_{j}}{s_{j}+t_{j}}$ to the entropy of $x'$, its contribution to that of $x''$ is at most $- \frac 1n \log 1/2 = \frac 1n$. It follows 
$$
\ent(x') - \ent(x'') \geq 
-\frac {s_{j}}n \log \frac{s_{j}}{s_{j}+t_{j}} -  \frac{s_{j}}{n}
=  \frac {s_{j}}n \log \frac{s_{j}+t_{j}}{s_{j}} - \frac{s_{j}}{n}.
$$
Since the rebalancing algorithm does not increase the entropy, we have
$\ent(x''') \leq \ent(x'')$, and hence
\begin{equation}
\label{eq:delta_ent_red}
\ent(x') - \ent(x''') \geq 
\frac {s_{j}}n \log \frac{s_{j}+t_{j}}{s_{j}} - \frac{s_{j}}{n}.
\end{equation}

Let us look at the difference $r_{j} - r_{j+1}$.
Every small vertex in $K$ becomes big at the end of the for-loop,
and all other vertices keep their status. Also, as we have already seen, the status of the vertices do not change during the rebalancing algorithm. Thus
\begin{equation}
\label{eq:delta_r_red}
r_{j} - r_{j+1} \geq s_{j}.
\end{equation}

Now assume $K$ is blue in $G_{j}(x')$. Since the algorithm gives priority
to good components that are red, and there is at least one such component if there is a red component,
it follows that every component of $G_{j}(x')$ is blue.
Then it can be checked that $x'_{v}$ is increased to $1$ for every small vertex $v$ in $K$
by the algorithm. (In fact, this is also true for every big vertex in $K$, though we will not use
that fact.)  We have
$$
\ent(x') - \ent(x'') \geq 
-\frac {s_{j}}n \log \frac{s_{j}}{s_{j}+t_{j}}
=  \frac {s_{j}}n \log \frac{s_{j}+t_{j}}{s_{j}}.
$$
Again, we have $\ent(x''') \leq \ent(x'')$, which implies
\begin{equation}
\label{eq:delta_ent_blue}
\ent(x') - \ent(x''') \geq 
\frac {s_{j}}n \log \frac{s_{j}+t_{j}}{s_{j}}.
\end{equation}

Here, there are no red vertices anymore. Thus
\begin{equation}
\label{eq:delta_r_blue}
r_{j} - r_{j+1} = 0.
\end{equation}

It follows from \eqref{eq:delta_ent_red} and \eqref{eq:delta_r_red} (in the red case)
and from \eqref{eq:delta_ent_blue} and \eqref{eq:delta_r_blue} (in the blue case)
that
\begin{equation}
\label{eq:delta_phi}
\phi_{j} - \phi_{j+1}
\geq s_{j}\log \frac{s_{j}+t_{j}}{s_{j}}.
\end{equation}

The right-hand side of \eqref{eq:delta_phi} 
can be bounded as follows:
\begin{equation}
\label{eqn:onemerge}
s_{j} \log \frac{s_{j}+t_{j}}{s_{j}} \geq \frac 12 s_{j} \log \frac{4t_{j}}{s_{j}}.
\end{equation}
This inequality follows from the fact that 
$\left(\frac{s_{j}+t_{j}}{s_{j}}\right)^{2}
= \frac{(s_{j}-t_{j})^{2}}{s_{j}^{2}} + \frac{4t_{j}s_{j}}{s_{j}^{2}} \geq \frac{4t_{j}}{s_{j}}$.

Now, let $q$ be the total number of comparisons done by the algorithm.
Using \eqref{eq:delta_phi}, \eqref{eqn:onemerge}, and Lemma~\ref{lem:HL_bound}, we obtain
$$
\phi_{1} = \sum_{j=1}^{k} (\phi_{j} - \phi_{j+1})
\geq \sum_{j=1}^{k}  s_{j} \log \frac{s_{j}+t_{j}}{s_{j}}
\geq \sum_{j=1}^{k} \frac 12 s_{j} \log \frac{4t_{j}}{s_{j}} 
\geq \frac{q}{2}.
$$
Thus
$$
q \leq 2\phi_{1} = 2n \ent(x) + 2 r_{1}.
$$

The number $r_{1}$ of small red vertices in $G_{1}(x)$ 
is equal to $n \ent(\tilde{x})$, where $\tilde x$ is the point defined
by letting $\tilde{x}_{v}= 1/2$ for every small red vertex $v$ in $G_{1}(x)$, and letting
$\tilde{x}_{v}= 1$ for every other vertex. 
The entropy of $\tilde{x}$ is at most the contribution 
of the red components in $G_{1}(x)$ to the entropy of $x$. The latter contribution
is in turn, by our assumption, at most $\ent(x)/2$. Therefore, 
$$
q \leq 2n \ent(x) + 2 r_{1} = 2n \ent(x) + 2n \ent(\tilde{x}) \leq 3n \ent(x),
$$
as claimed.
\end{proof}

\subsection{Putting Pieces Together: the Final Algorithm}
\label{sec:MUPI}

Our algorithm for the problem of merging under partial information is given below, see Algorithm \ref{alg:MUPI}. By combining Theorem \ref{th:c1} and Lemma \ref{lem:MUPIqc}, we obtain the following result. In the next section, we prove that the algorithm can be implemented so that its global complexity is $O(n^2 \log^2 n)$.

\begin{theorem}
Let $P$ be a poset covered by two disjoint chains $A$, $B$, and let $G = \bar{G}(P)$. Then Algorithm \ref{alg:MUPI} merges $A$ and $B$ in at most $6 \log e(P)$ comparisons.
\end{theorem}

\begin{algorithm}[h!]
\caption{Algorithm for Merging under Partial Information} \label{alg:MUPI}
\begin{algorithmic}[1]
\STATE \label{line:H(P)_width_2} compute a point $x \in \STAB(G)$ with $H(x)$ minimum
\IF{the contribution of the red components to $H(x)$ exceeds that of the blue components} 
\STATE exchange the chains $A$ and $B$
\ENDIF
\FOR{$v \in A \cup B$} 
\IF{$v$ is a cut-point} 
\STATE{copy $v$ at its final position in the chain $C$}
\ENDIF
\ENDFOR
\STATE call Algorithm \ref{alg:MUPI_core}
\RETURN $C$
\end{algorithmic}
\end{algorithm}

\subsection{Complexity}
\label{sec:MUPI_complexity}

In this section, we sketch an efficient implementation of the main steps of Algorithm \ref{alg:MUPI}, namely computing the entropy of a poset of width at most $2$ (line \ref{line:H(P)_width_2} of Algorithm \ref{alg:MUPI}), merging a pair of disjoint chains (line \ref{line:HL} of Algorithm \ref{alg:MUPI_core}, called by Algorithm \ref{alg:MUPI}), and updating after a merging (lines \ref{line:begin_for-loop1}--\ref{line:end_for-loop2} of Algorithm \ref{alg:MUPI_core}, called by Algorithm \ref{alg:MUPI}). 

There are some differences between the way the algorithms are described above, and the way they are implemented here: for the sake of efficiency, we sometimes change the order of some steps or use ways to accelerate some others. 

We start by briefly discussing the data structures used.
 
\paragraph*{Data Structures} 

The two chains $A = \{a_1,\ldots,a_{|A|}\}$ and $B = \{b_1,\ldots,b_{|B|}\}$ are kept in separate vectors which are never modified during the course of the algorithm (throughout, we assume $a_{i} \leqslant_P a_{i+1}$ for $i = 1, \ldots, |A|-1$ and $b_{j} \leqslant_P b_{j+1}$ for $j = 1, \ldots, |B|-1$). The output chain $C$ is a vector of size $n$, initialized arbitrarily. As soon as the `true' rank (that is, the rank in the linear order $\leqslant$) of a vertex is known, it is copied to the corresponding entry of $C$.

The data structure for the incomparability graph $G$ has two parts: a static part and a dynamic part. The dynamic part also contains information that allows us to monitor the evolution of the point $x \in \STAB(G)$, and in particular the components of $G(x)$.

The static part records the initial neighborhood of each vertex of $G$, as it is at the beginning of the algorithm. Because each of these neighborhoods is an interval of either $A$ or $B$, it suffices to record the indices of the first and last vertex within each interval. For instance, consider a vertex $v \in B$, and let $[a_i,a_j] := \{a_i,\ldots,a_j\}$ denote its neighborhood in $G$. Then we record the pair $(i,j)$. We allow $j = i-1$ when $v$ is a cut-point. (In this case, the rank of $v$ in the linear order $\leqslant$ is precisely its rank in the chain $B$, plus $i-1$.)

The dynamic part consists of the list of non-trivial components of $G$ and, for each such component of $G$, the list of non-trivial components of $G(x)$ contained in the corresponding component of $G$. The order of the components in each list is kept consistent with $P$. (Recall that $x\in \STAB^{*}(G)$ during the whole algorithm, and hence $P$ induces a linear ordering on the non-trivial components of $G(x)$ by Lemma~\ref{lem:int}.)  Trivial components (balanced or not) are not explicitly stored. 

Extra information is stored in the nodes of these lists. Consider a component $Z$ of $G$. Then, by Lemma \ref{lem:width_2_prop}, both $A \cap Z$ and $B \cap Z$ are intervals. We store the indices of the first and last vertices of each of these intervals, in the node for $Z$. This is used to implicitly maintain the neighborhood of each vertex of $G$: the current neighborhood of a vertex is the intersection of its initial neighborhood and of the component of $G$ that contains it.

In the node corresponding to a non-trivial component $K$ of $G(x)$, we store the indices of the first and last vertices of the intervals of $A \cap K$ and $B \cap K$. We also store the value of $x_u$ for some $u \in A \cap K$ and of $x_v$ for some $v \in B \cap K$ (because $K$ is a component of $G(x)$, the point $x$ is constant on both $A \cap K$ and $B \cap K$).

On the side, we maintain the list of unbalanced components of $G(x)$ (here, the order of the components in the list is arbitrary). Extra information is placed in the lists of components of $G(x)$ so that locating a given unbalanced component takes constant time. Similarly, we maintain the list of good components of $G(x)$ (the red components are systematically placed before the blue ones).

\paragraph*{Computing the Entropy}

In Appendix~\ref{app:bip_convex}, we prove the next result which implies that line \ref{line:H(P)_width_2} of Algorithm \ref{alg:MUPI} can be performed in $O(n^2 \log^2 n)$ time. This is due to the fact that, in virtue of Lemma \ref{lem:width_2_prop}, the incomparability graph of a poset of width at most $2$ is bipartite and biconvex, thus in particular bipartite and convex.

\begin{lemma}
\label{lem:H(P)_convex}
The entropy of an $n$-vertex convex bipartite graph $G$ can be computed in time $O(n^2 \log^2 n)$. 
\end{lemma}

\paragraph*{Rebalancing} 

Assume for now that the current point $x$ in Algorithm \ref{alg:rebalance} is such that $G(x)$ has no loose component. Then, processing an unbalanced component $K$ of $G(x)$ (see lines \ref{line:compute_slack}--\ref{line:add_slack} of Algorithm \ref{alg:rebalance}) can be done in constant time: it suffices to check the components of $G(x)$ that touch $K$, and perform the necessary updates. Letting $Z$ denote the component of $G$ that contains $K$, these components are the neighbors of $K$ in the list of components of $G(x)$ contained in $Z$. Thus there are at most two components to check. 

Now, if there were loose components in $G(x)$, we handle them separately
before calling Algorithm \ref{alg:rebalance}: These components are treated
simultaneously and in constant time during a specific ``updating'' phase right after the merging of the two chains $X$ and $Y$. This is explained below.

\paragraph*{Merging}

We implement the Hwang-Lin algorithm (line \ref{line:HL} of Algorithm \ref{alg:MUPI_core}, for a description see Section~\ref{sec:MUPI_core} or the original paper~\cite{HL72}) so that its complexity is proportional to the number of comparisons it performs, plus the number of cut-points discovered. This is possible because none of the vertices in $A$ or $B$ is moved. Each cut-point is copied to the output chain as soon as it is found (more details are given below).

\paragraph*{Updating After a Merging}

After the chains $X$ and $Y$ are merged, they are both split in at most three intervals: $X_1$ contains the vertices in $X$ that are ranked below all vertices in $Y$ in the merged chain, $X_2$ contains the vertices in $X$ that are ranked between two vertices in $Y$ in the merged chain and $X_3$ contains all the other vertices in $X$, that is, all those that are ranked above all vertices in $Y$. The intervals $Y_1$, $Y_2$ and $Y_3$ are defined similarly. Obviously, either $X_1$ or $Y_1$ is empty, and the same holds for $X_3$ and $Y_3$.

The vertices of the middle intervals $X_2$ and $Y_2$ become cut-points and are thus copied in the output chain $C$. Some extra vertices in $X_1$, $Y_1$, $X_3$ or $Y_3$ may also become cut-points. This is determined by inspecting the neighborhoods of at most four vertices in the components of $G(x)$ that touch the component $K$. More precisely, letting $J$ and $L$ denote the components of $G(x)$ that touch $K$ and directly precede or follow $K$, respectively (possibly, $J$ or $L$ is not defined), then we only have to inspect the last vertices of $J \cap A$ and $J \cap B$ and the first vertices of $L \cap A$ and $L \cap B$.

The above information, which describes the precise way in which components of $G$ and $G(x)$ change, can be obtained during the merging, essentially at no extra cost. Knowing it, we can update the list of components of $G$, and the lists of components of $G(x)$: the component of $G$ containing $K$ is typically split in two components, $K$ is deleted, and the components of $G(x)$ that touch the component $K$ are updated, as is explained in the next paragraph.
 
The vertices in $X_1 \cup Y_1$ that do not become cut-points (if any) are incorporated in the component $J$ (these vertices exactly correspond to the loose components of $G'(x)$ that touch $K$), and the vertices in $X_3 \cup Y_3$ that do not become cut-points (if any) are incorporated in the component $L$ (these vertices exactly correspond to the loose components of $G'(x)$ that touch $L$). 

Thus we do not implement lines \ref{line:begin_for-loop1}--\ref{line:rebalance_core} of Algorithm \ref{alg:MUPI_core} as is, but we rather process all loose components simultaneously, and then continue the rebalancing step normally. As said previously, a similar remark is in order for lines 
\ref{line:begin_for-loop2}--\ref{line:end_for-loop2} of Algorithm \ref{alg:MUPI_core}: we actually copy cut-points in the output chain as soon as possible.

\begin{figure}[ht]
\begin{center}
\includegraphics[scale=.6]{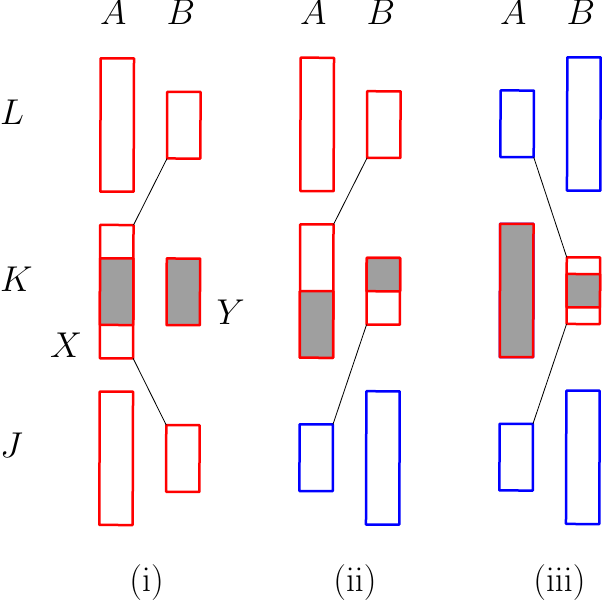}
\end{center}
\caption{Evolution of the components after a merging: three main cases.\label{fig:compupdate}}
\end{figure}

The possible evolution of the components after a merging is shown in Figure \ref{fig:compupdate}. We have illustrated three cases: (i) $J$ and $L$ have the same color as $K$, (ii) only $L$ has the same color as $K$, (iii) both $K$ and $L$ have different colors. The only edges with at least one endpoint in $K$ that may be present in $G'$ are displayed in the figure. Portions of the chains shown in gray depict vertices that become cut-points.

Because the number of operations, when the operations necessary for  merging pairs of chains or discovering cut-points is put aside, is linear in the number of components that $G(x)$ initially had, and each such operation takes constant time, we infer the following result, that is crucial to the next section.

\begin{lemma}
\label{lem:core_complexity}
Algorithm \ref{alg:MUPI_core} can be implemented so that its running time is $O(q)+O(n)$, where $q$ is the number of comparisons performed.
\end{lemma}

Lemmas \ref{lem:H(P)_convex} and \ref{lem:core_complexity} together imply that Algorithm \ref{alg:MUPI} can be implemented so that its running time is $O(n^2 \log^2 n)$.

\section{Reducing the Sorting Complexity}
\label{sec:reuse}

Recall that the preprocessing complexity is the number of operations performed before the first comparison, while the remaining operations account for the sorting complexity. Our goal in this section is to provide an algorithm whose sorting complexity is $O(\log e(P)) + O(n)$. By confining the entropy computation in the preprocessing phase, we are able to reuse the result of this preprocessing to sort any other instance with the same partial information. 

\begin{algorithm}[h!]
\caption{Sorting under Partial Information with reduced sorting complexity} \label{alg:SUPIpre}
\begin{algorithmic}[1]
\STATE \COMMENT{\textit{Phase 1 (preprocessing)}} 
\STATE find a maximum chain $A \subseteq P$
\STATE let $G$ be the bipartite spanning subgraph of $\bar{G}(P)$ 
containing all the edges between $A$ and $P-A$, and no other edge
\STATE compute a point $x\in \STAB(G)$ with $\ent(x)$ minimum
\STATE compute the function $f$ \COMMENT{\textit{(see below)}}
\STATE compute a greedy chain decomposition and the corresponding Huffman tree for the poset $P - A$
\STATE \COMMENT{\textit{Phase 2 (sorting)}} 
\STATE apply Algorithm~\ref{alg:merge_sort} to the poset $P - A$, yielding a chain $B$
\STATE \label{line:updateG} update the graph $G$ 
\STATE \label{line:updatecomp} compute the components of $G(x)$ and eliminate inlays by improving $x$
\STATE \label{line:loosevert} handle loose vertices of $G(x)$
\STATE \label{line:rebalance} rebalance $x$ using Algorithm~\ref{alg:rebalance}
\IF{the contribution of the red components to $\ent(x)$ exceeds that of the blue components} 
\STATE exchange the chains $A$ and $B$
\ENDIF
\STATE call Algorithm~\ref{alg:MUPI_core} 
\RETURN the resulting chain $C$
\end{algorithmic}
\end{algorithm}

The main idea of Algorithm~\ref{alg:SUPIpre} is to compute a minimum entropy point $x$ of a bipartite graph $G$ that can be defined before the sorting phase, solely on the basis of the initial partial information $P$. The following lemma shows that the entropy of $x$ provides enough information to guide
the sorting phase.

\begin{lemma}
Algorithm~\ref{alg:SUPIpre} is an algorithm for the problem of sorting under partial information with query complexity at most $15.09 \log e(P)$.
\end{lemma}
\begin{proof}
The proof is the same as that of Theorem~\ref{thm:mupired}, but where the query complexity of Algorithm~\ref{alg:MUPI_core} for merging under partial information is now at most $3n\ent (x)$ from Lemma~\ref{lem:MUPIqc}.

Since $G$ is a subgraph of $\bar{G}(P)$, its entropy is at most that of $\bar{G}(P)$. Combining this with Theorem~\ref{th:c1}, we get:
$$
3n\ent (x) \leq 3n\ent(\bar{P}) \leq 6\log e(P) .
$$
Hence, following the same reasoning as in Theorem~\ref{thm:mupired}, we obtain that Algorithm~\ref{alg:SUPIpre} has query complexity at most $(9.09 + 6)\log e(P) = 15.09 \log e(P)$.
\end{proof}

\subsection{Preprocessing}

The preprocessing phase involves computing the entropy of a convex bipartite graph $G$. This can be done in $O(n^2 \log^2 n)$ time, see Lemma \ref{lem:H(P)_convex}.

During this phase, we also compute the function $f$ that associates to each interval $[c,d]$ of the chain $A$, the maximum $M$ of $x_{u}$ over all $u \in [c, d]$, together with the largest interval $[c', d'] \subseteq [c, d]$ such that $x_{c'} = x_{d'}= M$. This can be done in $O(n^{2})$ time and space using a straightforward dynamic program.

\subsection{Sorting}

We now have to show that the sorting complexity is $O(\log e(P)) + O(n)$. Thus all operations of the sorting phase of Algorithm~\ref{alg:SUPIpre} have to be implemented with a $O(n)$ overhead. The main issues in that respect are the complexities of lines~\ref{line:updateG}--\ref{line:rebalance}.

\paragraph*{Updating the Graph $G$}
In line~\ref{line:updateG} of the algorithm, we modify $G$ so that it becomes the incomparability graph of the partial information we have right after the chain $B$ has been computed. This is an update, in the sense that the new graph $G$ will be a spanning subgraph of the old one.

To perform this update in linear time, we make use of the structural observations of Lemma~\ref{lem:width_2_prop}. In particular, property (iii) of this lemma allows us to recover the incomparability interval of every vertex in $A$ ($B$) by scanning twice the chain $B$ ($A$, respectively). More precisely, we first scan the chain $A$ from bottom to top and find, for each vertex in $A$, the lower endpoint of its incomparability interval in $B$. These endpoints are increasing; hence, this does not require any backtracking in $B$. A second scanning from top to bottom yields the upper endpoints. The incomparability intervals for vertices in $B$ are computed similarly. Therefore, the whole graph $G$ can be computed in $O(n)$ time.

\paragraph*{Finding the Components of $G(x)$}
At line~\ref{line:updatecomp} of the algorithm, we aim at computing the components of $G(x)$ and encoding them in the data structure described in Section~\ref{sec:MUPI_complexity}. During this step, we also modify the point $x$ so that inlays in $G(x)$ are avoided. These weight modifications consist simply in increasing $x_{v}$ for some vertices $v$ (without loosing feasibility); hence the entropy of $x$ can only decrease during this step.

We proceed by scanning $B$ from bottom to top. First, for every vertex $v$ in $B$, we apply the function $f$ on the incomparability interval $[c,d]$ of $v$, and obtain the corresponding maximum $M(v)$ and an interval $[c', d'] \subseteq [c,d]$. If $x_{v} + M(v) = 1$, then $v, c'$, and $d'$ belong to the same component of $G(x)$ (possibly $c' = d')$. We save such ``tight'' intervals in a list, together with the corresponding vertices $v$, and forget about the intervals that are not tight.

Next, we compute in linear time the union of all tight intervals in the list (by scanning the list once and merging consecutive intervals when they intersect). This results in a collection of disjoint intervals $[u_{1},u'_{1}], \dots, [u_{\ell},u'_{\ell}]$ of $A$. For each such interval $[u_{i},u'_{i}]$, we can compute in constant time the smallest vertex $v_{i} \in B$ and largest vertex $v_{i}' \in B$ such that the tight intervals of $v_{i}$ and $v_{i}'$ are included in $[u_i,u'_i]$ (possibly $v_{i} = v'_{i}$). Observe that the intervals $[v_{1},v'_{1}], \dots, [v_{\ell},v'_{\ell}]$ of $B$ are also disjoint. Also, it can be checked that for every $i$, we have $x_{v_{i}} = x_{v'_{i}} = 1 - x_{u_{i}} = 1 - x_{u'_{i}}$. Moreover, for every $v \in [v_{i}, v_{i}']$, we have $x_{v} \leq x_{v_{i}} = x_{v'_{i}}$. Similarly, for every $u \in [u_i,u'_i]$, we have  $x_{u} \leq x_{u_{i}} = x_{u'_{i}}$. Thus, for every $i \in \{1, \dots, \ell\}$, we can safely update the point $x$ as follows: we increase $x_{v}$ to $x_{v_{i}}$ for every $v \in [v_{i}, v_{i}']$, and similarly
increase $x_{u}$ to $x_{u_{i}} = 1 - x_{v_{i}}$ for every $u \in [u_{i},u'_{i}]$. 

The non-trivial components of $G(x)$ (with $x$ updated as above) are exactly given by the collection $[u_{i},u_{i}'] \cup [v_{i}, v_{i}']$ for $i \in \{1, \dots, \ell\}$. Notice that the components of $G(x)$ are now free of inlays (as defined in Section~\ref{sec:locimp}).

\paragraph*{Handling Loose Vertices}

It remains to process vertices that are not incident to any tight edge in $G(x)$, but that are not cut-points either (line~\ref{line:loosevert} of the algorithm). We again scan $B$ bottom-up, and for each loose vertex $v\in B$, check the weights associated with the non-trivial components touching the component $\{v\}$. There are at most two such components. We also apply the function $f$ to the interval of vertices in $A$ strictly between those components, and within the bounds of the incomparability interval of $v$. This allows us to determine in constant time a slack value by which we can increase $x_{v}$. The vertex $v$ may now be included in a previously defined component of $G(x)$, or form a new component with loose vertices of $A$.

Afterwards, the remaining loose vertices of $A$ can be eliminated in a similar fashion. For those, however, no new component can be created, as there are no loose vertices remaining in $B$.

It can be checked that the above procedure involves only a constant number of operations per vertex. Therefore, the complexity of lines~\ref{line:updatecomp}--\ref{line:loosevert} is $O(n)$.

\paragraph*{Rebalancing}
The rebalancing step (line~\ref{line:rebalance}) involves Algorithm~\ref{alg:rebalance}. At every iteration of this algorithm, the number of components of $G(x)$ decreases, hence there can be at most a linear number of iterations. Every iteration takes constant time using the data structure described in Section~\ref{sec:MUPI_complexity} for the components (this data structure can be used because $G(x)$ has no loose components). Thus the complexity of the rebalancing step is $O(n)$ as well.

\section*{Acknowledgments and a Final Remark}

We thank an anonymous referee for the numerous insightful comments and pointers to relevant references. 

In particular, the referee pointed out to us the possibility of an efficient implementation of Linial's algorithm for merging under partial information~\cite{L84}. The latter algorithm takes advantage of the fact that computing the number of linear extensions of partial orders $P$ that can be covered by two disjoint chains (the ones we deal with in Section~\ref{sec:MUPIsec}) can be done in polynomial time. Hence we can find an efficient query, that is, a query ``is $v_i\leq v_j$?'' that splits the space of linear extensions as evenly as possible, in polynomial time. Linial suggests the use of determinants to count linear extensions, which is likely to be inefficient. It is however possible to improve on this and compute those numbers via a simple dynamic program over the downsets of $P$. When a query is answered, it is possible to update the dynamic programming table locally, so as to reuse as much information as possible from the previous steps. In order to avoid the problems of dealing with huge numbers, the arithmetic operations can be performed with limited precision. 

It is likely that this algorithm would be competitive with the solution proposed here as Algorithm 7, and conceptually much simpler. It does not have the property, however, to have separated preprocessing and sorting phases, which is the main point of the current developments and Algorithm 7.

\bibliographystyle{plain}
\bibliography{supi}

\appendix

\section{Greedy Chain Decompositions}
\label{app:greedy}

Any given poset $P$ can be canonically decomposed into ``levels''. To construct this decomposition, we find the set $L_1$ of minimal elements of $P$ (that is, the elements of $P$ without predecessor), then set $L_2$ of minimal elements of $P-L_1$, and continue likewise until we find a set $L_h$ such that $P-L_1-\cdots-L_h$ is empty. The set $L_i$ is the $i$th {\sl level\/} of $P$, and $h$ is the {\sl height\/} of $P$. By construction, every element of $L_i$ has a predecessor in $L_{i-1}$, for $i=2,\ldots,h$. Thus $P$ contains a chain of size $h$. Because each level is an antichain, the maximum size of a chain in $P$ is precisely $h$.

The levels of a poset $P$ of order $n$ can be found in time $O(n^2)$. If, while constructing the levels, we record for each vertex in a level $L_i$ with $i \ge 2$ one of its predecessors in the previous level $L_{i-1}$, a maximum chain of $P$ can be then found in time $O(h)$.

\begin{proposition}
There is a $O(n^{2.5})$ algorithm finding a greedy chain decomposition of any poset of order $n$.
\end{proposition}

\begin{proof}
We assume we know all the relations of $P$. If needed, we compute a transitive closure in time $\widetilde{O}(n^\omega)$, where $\omega$ is any real such that any two $n \times n$ matrices can be multiplied by performing $O(n^\omega)$ arithmetic operations, e.g., $\omega =  2.376$.
%Adleman, Booth, Preparata and Ruzzo give O(n^\omega \log n \log \log \log n \log \log \log \log n)

While the height of $P$ exceeds $\sqrt{n}$, we repeat the following steps: build the decomposition of $P$ into levels from scratch, find a maximum chain $C$ in $P$, record $C$ and remove $C$ from $P$. This first phase takes $O(\sqrt{n}\,n^2) = O(n^{2.5})$ time.

Now the height of $P$ is at most $\sqrt{n}$. We continue as before except that instead of rebuilding the levels each time from scratch, we update them. To this end, we maintain for each element $v$ of $P$ a {\sl table of predecessors}. Suppose $v$ lies in level $L_i$. Then the $j$th entry of the table gives the list of predecessors of $v$ lying $j$ levels down, in level $L_{i-j}$.

Updating the levels is done as follows. First, for each element $u$ of the chain $C$, we delete $u$ from $P$ and update the table of predecessors of every successor of $u$. We mark every element $v \in P - C$ such that the first component of the predecessor table for $v$ becomes empty. Second, for $i = 1, \dots, h$, we process the $i$th level $L_i$: For each element $u$ that is marked, we determine the minimum index $j$ such that the $j$th component of the predecessor table for $u$ is non-empty, move $u$ in level $L_{i-j}$, update the predecessor table for $u$ and the predecessor table of every successor $v$ of $u$. Again, we mark every element $v$ such that the first component of the predecessor table for $v$ becomes empty.

In order to analyze the algorithm, we assign to each relation of $P$ a ``score''. The {\sl score\/} of $u \leqslant_P v$ is $i + j$, where $i$ and $j$ are the indices of the levels containing $u$ and $v$, respectively. Initially, the score of each relation is $O(\sqrt{n})$. Each time a relation is considered, its score is decreased by at least one. Hence, a given relation is considered $O(\sqrt{n})$ times through all the updates. Thus, the second phase of the algorithm also takes $O(n^2\sqrt{n}) = O(n^{2.5})$ time.

Therefore, a greedy chain decomposition can be found in $O(n^{2.5}$) time.
\end{proof}

\section{Computing the Entropy of Convex Bipartite Graphs}
\label{app:bip_convex}

\begin{proof}[Proof of Lemma \ref{lem:H(P)_convex}]
Let $A, B$ denote a bipartition of the vertices of $G$. Without loss of generality, $G$ is $A$-convex, that is, there is a linear ordering on the vertices in $A$ such that the neighborhood of every vertex of $B$ is an interval in $A$. 

We explain how to implement one iteration of the method of K\"orner and Marton~\cite{KM88} described in Section \ref{sec:H_bip}. As previously, we denote by $G'$ the current graph, and by $A'$, $B'$ its current bipartition. Thus $G'$ is $A'$-convex.

Vertices in $A'$ that are isolated in $G'$ are dealt with first and separately. Thus, we may assume that no vertex in $A'$ is isolated in $G'$. Similarly, we may assume that $A'$ is nonempty. 

First, the algorithm determines the maximum ratio \eqref{eq:bip_ratio} achievable by a subset $A_i \subseteq A'$, by bisection. Let
$$
\rho := \beta / \alpha
$$ 
denote the guessed ratio, with $\beta, \alpha \in \{1,\ldots,n\}$. Since there are $O(n^2)$ possible ratios, the number of guesses is $O(\log n)$. Next, we prove that we can decide in $O(n \log n)$ time whether there exists a subset $A_i \subseteq A'$ whose ratio is larger than $\rho$, or whether no such subset exists.

Consider the network $D$ obtained from $G'$ by directing all its edges from $A'$ to $B'$, adjoining a source vertex $s$ sending a directed edge to each vertex of $A'$, and a sink vertex $t$ receiving a directed edge from each vertex of $B'$. The capacities of the directed edges incident to $s$ (resp.\ $t$) are set to $\alpha$ (resp.\ $\beta$). The capacities of the other directed edges are set to $\infty$. Because the $s$--$t$ cut defined by $\{s\} \cup A'$ has capacity $\alpha |A'|$, two cases can occur: either the minimum capacity of a $s$--$t$ cut in $D$ equals $\alpha |A'|$ ({\sl case (i)\/}), or it is less than $\alpha |A'|$ ({\sl case (ii)\/}). 

We claim that there exists a subset $A_i \subseteq A'$ such that the ratio \eqref{eq:bip_ratio} is larger than $\rho$ if and only if case (ii) arises. Indeed, if such a subset $A_i$ exists then the capacity of the cut defined by $\{s\} \cup A_i \cup N_{G'}(A_i)$ equals $\alpha(|A'| - |A_i|) + \beta |N_{G'}(A_i)|$, which is less than $\alpha |A'|$ because \eqref{eq:bip_ratio} is larger than $\beta / \alpha$. Conversely, if case (ii) arises then consider a minimum $s$--$t$ cut defined by $\{s\} \cup X \cup Y$, where $X \subseteq A'$ and $Y \subseteq B'$. By minimality of the cut, it follows that $Y = N_{G'}(X)$.
Because the capacity of the cut is less than $\alpha |A'|$, we conclude that
$$
\frac{|X|}{|N_{G'}(X)|} > \frac{\beta}{\alpha}.
$$
The claim follows. By the max-flow min-cut theorem, case (ii) arises if and only if the maximum value of a $s$--$t$ flow in $D$ is strictly smaller than $\alpha |A'|$. 

Computing a maximum $s$--$t$ flow in $D$ amounts to computing a maximum $b$-matching in the convex bipartite graph $G'$, where the weights of the vertices are defined as $b_u := \alpha$ whenever $u \in A'$ and $b_v := \beta$ whenever $v \in B'$. This can be done in $O(n \log n)$ time by adapting Glover's algorithm for computing a maximum matching in a convex bipartite graph \cite{G67} to the weighted case, and using a heap for storing vertices of $B'$. 

Second, once the maximum possible value of the ratio \eqref{eq:bip_ratio} is determined, we seek a maximizer $A_i \subseteq A'$. This amounts to converting the last maximum $s$--$t$ flow computed during the bisection into a minimum $s$--$t$ cut. Because case (i) arises, the value of the flow equals $\alpha |A'|$. Hence, all the directed edges incident to $s$ are saturated. If all the directed edges incident to $t$ are also saturated, then $\alpha |A'| = \beta |B'|$ and we may take $A_i := A'$. Otherwise, we perform a BFS from $t$ in an auxiliarly network obtained from $D$ by deleting all saturated directed edges, reversing all non-saturated edges (in particular, all the edges from $A$ to $B$) and adding all the directed edges $(u,v)$ from $A$ to $B$ that carry a nonzero flow. Because $G'$ is $A'$-convex and the support of the maximum $s$--$t$ flow in $D$ constructed by Glover's algorithm is of size $O(n)$, we can perform the BFS in $O(n \log n)$ time. We then define $A_i$ as the vertices of $A'$ that cannot be reached from $t$ in the auxiliary network. The lemma follows.
\end{proof}

\end{document}